\documentclass[runningheads]{llncs}
\usepackage[T1]{fontenc}
\usepackage{graphicx}
\usepackage{xspace}
\usepackage{tikz}
\usetikzlibrary{arrows, calc, shapes,petri}
\usepackage{pgfplots} 
\pgfplotsset{compat = newest}
\usepackage{paralist}
\usepackage{amssymb}
\usepackage{amsmath}
\usepackage{todonotes}
\usepackage{thmtools, thm-restate}
\usepackage{hyperref}
\usepackage{tabularx}
\usepackage{booktabs}
\usepackage{rotating}
\usepackage{color}

\usepackage{amssymb,bm}
\usepackage{pifont}% http://ctan.org/pkg/pifont
\newcommand{\cmark}{\ding{51}}%
\newcommand{\xmark}{\ding{55}}
\newcommand{\smark}{\ensuremath{\bm{\thicksim}}}

\newcolumntype{C}{>{\centering\arraybackslash}X}
\newcolumntype{R}{>{\raggedleft\arraybackslash}X} 
\newcolumntype{L}{>{\raggedright\arraybackslash}X} 
\newcolumntype{P}[2]{%
  >{\begin{turn}{#1}\begin{minipage}{#2}\raggedright}l%
  <{\end{minipage}\end{turn}}%
}
\newcolumntype{u}[1]{%
  >{\begin{minipage}{#1}\raggedleft}l%
  <{\end{minipage}}%
}

\newcommand{\mc}[1]{\mathcal{#1}}
\newcommand{\m}[1]{\mathsf{#1}}

\renewcommand{\log}{L}
\newcommand{\tup}[1]{\langle{#1}\rangle}
\newcommand{\events}{E}
\newcommand{\objects}{\mathcal O}
\newcommand{\otypes}{\Sigma}
\newcommand{\activities}{\mc A}
\newcommand{\timestamps}{\mathbb T}
\newcommand{\proj}{\pi}
\newcommand{\projact}{\proj_{\mathit{act}}}
\newcommand{\projobj}{\proj_{\mathit{obj}}}
\newcommand{\projtime}{\proj_{\mathit{time}}}
\newcommand{\projtrace}{\proj_{\mathit{trace}}}
\newcommand{\otype}{\vartype}

\newcommand{\pre}[1]{\bullet{#1}}
\newcommand{\post}[1]{{#1}\bullet}
\newcommand{\logtrace}{\mathbf e}
\newcommand{\goto}[1]{\mathrel{\raisebox{-2pt}{$\xrightarrow{#1}$}}}

\newcommand{\GG}{\mc G} % graphs
\newcommand{\LL}{\mc L} % language
\newcommand{\NN}{\mc N} % nets
\newcommand{\secref}[1]{Sec.~\ref{sec:#1}}

\newcommand{\exaref}[1]{Ex.~\ref{exa:#1}}
\newcommand{\defref}[1]{Def.~\ref{def:#1}}
\newcommand{\remref}[1]{Rem.~\ref{rem:#1}}
\newcommand{\lemref}[1]{Lem.~\ref{lem:#1}}
\newcommand{\figref}[1]{Fig.~\ref{fig:#1}}
\newcommand{\tabref}[1]{Tab.~\ref{tab:#1}}

\tikzstyle{place}=[draw, circle, inner sep=1.5pt, line width=.7pt, scale=.8, minimum width=6mm]
\tikzstyle{trans}=[draw, rectangle, inner sep=1.5pt, line width=.7pt, scale=.8, minimum width=6mm, minimum height=6mm, fill=gray!10]
\tikzstyle{arc}=[draw, ->, line width=.5pt]
\tikzstyle{fatarc}=[draw, ->, line width=.5pt, double]
\tikzstyle{action}=[scale=.6]
\tikzstyle{starttoken}=[regular polygon, regular polygon sides=3,minimum width=2mm,fill=black,inner sep=0pt,rotate=30]
\tikzstyle{endtoken}=[regular polygon, regular polygon sides=4,minimum width=2mm,fill=black,inner sep=0pt]
\colorlet{ocolor}{blue!70!green!20}
\colorlet{icolor}{yellow!90!red!20}
\colorlet{oicolor}{green!90!blue!20}
\tikzstyle{oplace}=[place,fill=ocolor]
\tikzstyle{otrans}=[trans,fill=ocolor]
\tikzstyle{iplace}=[place,fill=icolor]
\tikzstyle{itrans}=[trans,fill=icolor]
\tikzstyle{oiplace}=[place,fill=oicolor]
\tikzstyle{idplace}=[place,fill=red!30] 
\tikzstyle{dplace}=[place,fill=magenta!20]             
\tikzstyle{oidplace}=[place,fill=black!20]
\tikzstyle{oitrans}=[trans,fill=oicolor]
\tikzstyle{mixtrans}=[trans,shading = axis,rectangle, left color=ocolor, right color=icolor,shading angle=135]
\tikzstyle{oimixtrans}=[trans,shading = axis,rectangle, left color=ocolor, right color=oicolor,shading angle=135]
\tikzstyle{insc}=[scale=.8]
\tikzstyle{splitme}=[rectangle split, rectangle split horizontal,rectangle split parts=2]
\tikzstyle{log}=[fill=white]
\tikzstyle{model}=[fill=gray!10]

\newcommand{\objtuples}{\vec\objects}

\newcommand{\dom}{\mathit{dom}}

\newcommand{\vartype}{\funsym{type}}

\newcommand{\eventlog}{L}
\renewcommand{\log}{{\mathit{log}}}
\renewcommand{\mod}{{\mathit{mod}}}

\newcommand{\cost}{{\mathit{cost}}}
\newcommand{\run}{\rho}

\newcommand{\moves}{\mathit{moves}}
\newcommand{\SKIP}{{\gg}}
\newcommand{\restrlog}{|_{\mathit{log}}}
\newcommand{\restrmod}{|_{\mathit{mod}}}

\newcommand{\nuvarset}{\Upsilon}
\newcommand{\listvarset}{\varset_{\mathit{list}}}

\newcommand{\setsym}[1]{\mathit{#1}}
\newcommand{\places}{P}
\newcommand{\inflow}{F_{in}}
\newcommand{\outflow}{F_{out}}

\newcommand{\transitions}{T}

\newcommand{\varset}{\mathcal{V}}

\newcommand{\funsym}[1]{\mathtt{#1}}

\newcommand{\coloring}{\funsym{color}}
\newcommand{\set}[1]{\{#1\}}
\newcommand{\allvars}{\mathcal{X}}

\newcommand{\invars}[1]{\setsym{\setsym{vars}_{in}}(#1)}
\newcommand{\outvars}[1]{\setsym{\setsym{vars}_{out}}(#1)}

\newcommand{\vars}[1]{\setsym{vars}(#1)}
\newcommand{\colors}{\mathcal C}
\newcommand{\Dom}{dom}
\newcommand{\marking}{M}

\newcommand{\listtype}[1]{[{#1}]}
\newcommand{\tracenet}{T}

\newcommand{\eid}[1]{\mathtt{\#}_{#1}}

\newcommand{\mynet}{OPID\xspace}
\newcommand{\mynets}{OPIDs\xspace}
\newcommand{\anet}{an \mynet}

\newcommand{\yices}{\textsf{Yices}\xspace}

\newcommand{\cocomot}{\textsf{CoCoMoT}\xspace}
\newcommand{\thetool}{\textsf{oCoCoMoT}\xspace}
\newcommand{\skippath}[1]{\to_{#1}^{\SKIP}}
\newcommand{\eqn}{\,{=}\,}

\newcommand{\SW}[1]{{#1}}

\begin{document}
\title{Object-Centric Conformance Alignments\\ with Synchronization (Extended Version)\thanks{Gianola was partially supported by national funds through
FCT, Funda\c{c}\~{a}o para a Ci\^{e}ncia e a Tecnologia, under projects
UIDB/50021/2020 (DOI:10.54499/UIDB/50021/2020). Montali and Winkler acknowledge the UNIBZ project ADAPTERS and the PRIN MIUR project PINPOINT Prot.~2020FNEB27.}}
\titlerunning{Object-Centric Conformance Alignments with Synchronization}
% If the paper title is too long for the running head, you can set
% an abbreviated paper title here
%
% \author{First Author\inst{1}\orcidID{0000-1111-2222-3333} \and
% Second Author\inst{2,3}\orcidID{1111-2222-3333-4444} \and
% Third Author\inst{3}\orcidID{2222--3333-4444-5555}}
\author{Alessandro Gianola\inst{1} \and Marco Montali\inst{2} \and Sarah Winkler\inst{2}}
%
% \authorrunning{A. Gianola et al.}
% First names are abbreviated in the running head.
% If there are more than two authors, 'et al.' is used.
%
% \institute{Princeton University, Princeton NJ 08544, USA \and
% Springer Heidelberg, Tiergartenstr. 17, 69121 Heidelberg, Germany
% \email{lncs@springer.com}\\
% \url{http://www.springer.com/gp/computer-science/lncs} \and
% ABC Institute, Rupert-Karls-University Heidelberg, Heidelberg, Germany\\
% \email{\{abc,lncs\}@uni-heidelberg.de}}
\institute{INESC-ID/Instituto Superior T{\'e}cnico, Universidade de Lisboa, Portugal \email{alessandro.gianola@tecnico.ulisboa.pt} \and Free University of Bozen-Bolzano, Italy \\ \email{\{montali,winkler\}@inf.unibz.it}}
\maketitle              % typeset the header of the contribution
\begin{abstract}
Real-world processes operate on objects that are inter-de\-pen\-dent. To accurately reflect the nature of such processes, object-centric process mining techniques are needed, notably conformance checking.
However, while the object-centric perspective has recently gained traction, few concrete process mining techniques have been presented so far.
Moreover, existing approaches are severely limited in their abilities to keep track of object identity and object dependencies. Consequently, serious problems in event logs with object information remain undetected.
This paper, presents a new formalism 
that combines the key modelling features of two existing approaches, notably the ability of object-centric Petri nets to capture one-to-many relations and the ability of Petri nets with identifiers to compare and synchronize objects based on their identity.
% that integrates the object-centric Petri nets from work by van der Aalst and Berti, % (2020)
% with identifiers and their explicit manipulation as in Petri nets with identifiers (PNIDs).
We call the resulting formalism \emph{object-centric Petri nets with identifiers}, and define alignments and the conformance checking task for this setting.
We propose a conformance checking approach for such nets based on an encoding in satisfiability modulo theories~(SMT), and illustrate how it 
%can be effectively used 
serves
to effectively overcome shortcomings of earlier work.
To assess its practicality, we 
%perform an evaluation 
evaluate it
on data from the literature.
\keywords{BPM  \and conformance checking \and object-centric processes \and object-centric process mining \and SMT.}
\end{abstract}

\section{Introduction}
\label{sec:intro}

In information systems engineering, business/work processes are classically captured with a case-centric approach: every process instance focuses on the evolution of a main case object (such as an order, a claim, a patient), in isolation  to other case objects. This approach falls short in a variety of real-life, so-called \emph{object-centric processes}, where multiple objects are co-evolved depending on their mutual, many-to-many and one-to-many relations. 
A prime scenario witnessing this intricacy, which we use as main motivation throughout the paper, is the one of order-to-delivery processes, where an order consists of multiple items.
% Sarah: this is not what happens in our example
%where the items contained in different customer orders may be recombined and shipped inside multiple packages, each of which may contain items from distinct orders. 

The need of capturing this class of processes has emerged 
in modelling, verification and enactment \cite{PWOB19,AKMA19,Fahland19,SnSW21,GhilardiGMR22}. In particular, a number of different Petri net-based formalisms have been developed to deal with object-centricity, where one can distinguish: 
\begin{inparaenum}[\itshape (i)]
\item implicit object manipulation as in synchronous proclets~\cite{Fahland19} and
object-centric Petri nets~\cite{AalstB20}; and
\item explicit object manipulation as in ~\cite{WerfRPM22,PWOB19,GhilardiGMR22}, %Petri nets with identifiers~\cite{WerfRPM22} (at the basis of even richer approaches such as \cite{PWOB19,GhilardiGMR22}), 
which extend Petri nets with names \cite{RVFE10} with the ability of manipulating tuples of identifiers, to % - essential to deal with 
handle object relations.
\end{inparaenum}

Even more vehemently, the same call for object-centricity has  been recently advocated in process mining \cite{Aalst23,BeMA23}. %In fact, 
Many information systems (e.g., ERP systems \cite{BPRA21,CJKM23}) store event data related to several objects without any explicit reference to a single case object. Artificially introducing or
selecting one such object as the case %one 
during the extraction phase leads to misleading process mining outputs, in particular when dealing with convergence (events that simultaneously operate over multiple objects) and divergence (concurrent flows of related objects) \cite{Aalst19}. 

This paradigm shift calls for novel, object-centric process mining techniques, able to cope with 
% object-centric process models such as those mentioned above, which call for full consideration of the following features:
process models that fully support the following features:
\begin{compactenum}
\item activities creating and manipulating objects, as well as one-to-one and one-to-many relations among them---for example, the insertion of multiple items in an order, or the split of an order into multiple packages;\footnote{Many-to-many relations are typically reified into corresponding one-to-many relations, which is essential to properly handle synchronization constraints \cite{Fahland19,AKMA19,GhilardiGMR22}.}
\item concurrent flows evolving related objects separately---for example, the concurrent evolution of package shipments and order notifications;
\item object-aware synchronization points dictating that an object $o$ can flow through an activity only if some (\emph{subset synchronization}) or all (\emph{exact synchronization}) objects related to $o$ simultaneously flow as well --- e.g., enabling order billing only when some/all its packages have been delivered.
\end{compactenum}
Interestingly, while all aformentioned formalisms cope with the first two aspects, object-aware synchronization is completely missing in object-centric Petri nets \cite{AalstB20}, subset synchronization is tackled by Petri nets with identifiers \cite{WerfRPM22}, while synchronous proclets \cite{Fahland19} handle both subset and exact synchronization.

Unfortunately, the existing, few techniques for object-centric process discovery \cite{AalstB20} and conformance checking \cite{LissAA23} (two cornerstone problems in process mining) %only tackle object-centric Petri nets, and 
are %thus 
unable to track the identity of objects and their relations, and in turn to deal with synchronization. % when computing alignments. 
In conformance checking, which is the focus of %the present 
this paper, this means that serious deviations between the recorded and the expected behaviour remain undetected, as shown in the following example: %. The following example illustrates one such case.

\begin{example}
\label{exa:intro}
Consider the order process specified as an \emph{object centric Petri net}~\cite{AalstB20} in \figref{intro}~(a), with sorts \emph{order} (blue) and \emph{product} (yellow).
The following (partial) log uses two orders $o_1$, $o_2$ and two product items $p_1$, $p_2$. Intuitively, it should not conform to the process, since the items are shipped with the wrong order:
{

\centering
\begin{footnotesize}
% \begin{align*}
$
\begin{array}{rl}
&(\m{place\ order}, \{o_1,p_1\}), (\m{payment}, \{o_1\}), (\m{pick\ item}, \{p_1\}),
(\m{place\ order}, \{o_2,p_2\}), \\
&(\m{payment}, \{o_2\}), (\m{pick\ item}, \{p_2\}),  \notag 
(\m{ship}, \{o_1,p_2\}), (\m{ship}, \{o_2,p_1\})
\label{eq:log}
\end{array}
$
% \end{align*}
\end{footnotesize}
} \par
% \vspace{-7mm}

\noindent
However, according to                                                                                                                                                                                                                                                                                                                                                                                                                                                                                                                                                                                                                                                                                                                                                                                                                                                                                                                                                                                                                                                                                                                                                                                                                                                                                                                                                                                                                                                                                                                                                                                                                                                                                                                                                                                                                                                                                                                                                                                                                                                                                                                                                                                                                                                                                                                                                                                                                                                                                                                                                                                                                                                                                                                                                                                                                                                                                                                                                                                                                                                                                                                                                                                                                                                                                                                                                                                                                                                                                                                                                                                                                                                                                                                                                                                                                                                                                                                                                                                                                                                                                                                                                                                                                                                                                                                                                                                                                                                                                                                                                                                                                                                                                                                                                                                                                                                                                                                                                                                                                                                                                                                                                                                                                                                                                                                                                                                                                                                                                                                                                                                                                                                                                                                                                                                                                                                                                                                                                                                                                                                                                                                                                                                                                                                                                                                                                                                                                                                                                                                                                                                                                                                                                                                                                                                                                                                                                                                                                                                                                                                                                                                                                                                                                                                                                                                                                                                                                                                                                                                                                                                                                                                                                                                                                                                                                                                                                                                                                                                                                                                                                                                                                                                                                                                                                                                                                                                                                                                                                                                                                                                                                                                                                                                                                                                                                                                                                                                                                                                                                                                                                                                                                                                                                                                                                                                                                                                                                                                                                                                                                                                                                                                                                                                                                                                                                                                                                                                                                                                                                                                                                                                                                                                                                                                                                                                                                                                                                                                                                                                                                                                                                                                                                                                                                                                                                                                                                                                                                                                                                                                                                                                                                                                                                                                                                                                                                                                                                                                                                                             {the approach of~\cite{LissAA23}, this log has an alignment with cost 0 (cf. 
\figref{intro}~(b)), i.e., the mismatch is not detected.
Essentially, the problem is that their formalism does not keep track of object identity for synchronization.}
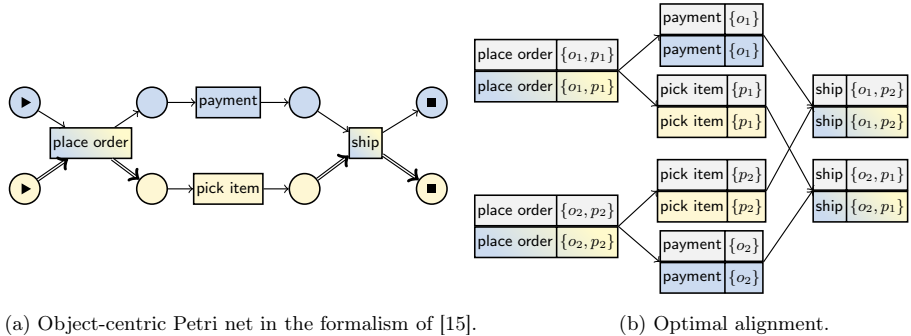
\begin{figure}[t]
\resizebox{\textwidth}{!}{
\begin{tikzpicture}[node distance=25mm]
\node[oplace] (o0) {} node[starttoken] {};
\node[oplace,right of=o0] (o1) {};
\node[otrans,right of=o1, xshift=-10mm] (ot1) {$\m{payment}$};
\node[oplace,right of=ot1, xshift=-10mm] (o2) {};
\node[oplace,right of=o2] (o3) {} node[endtoken] at (o3) {};
\node[iplace, below of=o0, yshift=8mm] (i0) {} node[starttoken] at (i0) {};
\node[iplace,right of=i0] (i1) {};
\node[mixtrans,below of=o1,yshift=17mm,xshift=-12mm] (po) {$\m{place\ order}$};
\node[itrans,right of=i1, xshift=-10mm] (it1) {$\m{pick\ item}$};
\node[iplace,right of=it1, xshift=-10mm] (i2) {};
\node[iplace,right of=i2] (i3) {} node[endtoken] at (i3) {};
\node[mixtrans,below of=o2,yshift=17mm,xshift=12mm] (ship) {$\m{ship}$};
\draw[arc] (o0) -- (po);
\draw[arc] (po) -- (o1);
\draw[arc] (o1) -- (ot1);
\draw[arc] (ot1) -- (o2);
\draw[arc] (o2) -- (ship);
\draw[arc] (ship) -- (o3);
\draw[fatarc] (i0) -- (po);
\draw[fatarc] (po) -- (i1);
\draw[arc] (i1) -- (it1);
\draw[arc] (it1) -- (i2);
\draw[fatarc] (i2) -- (ship);
\draw[fatarc] (ship) -- (i3);
\begin{scope}[xshift=82mm,yshift=5mm,xscale=1.3, yscale=.7]
\tikzstyle{splitme}=[rectangle split, rectangle split horizontal,rectangle split parts=2]
\node[mixtrans, anchor=north,splitme] (pom) {$\m{place\ order}$\nodepart{two}$\{o_1,p_1\}$};
\node[trans, anchor=south,splitme] (pot) {$\m{place\ order}$\nodepart{two}$\{o_1,p_1\}$};
\node[otrans, anchor=north,splitme] (pm) at (2,.8) {$\m{payment}$\nodepart{two}$\{o_1\}$};
\node[trans, anchor=south,splitme] (pt) at (2,.8) {$\m{payment}$\nodepart{two}$\{o_1\}$};
\node[itrans, anchor=north,splitme] (pi1m) at (2,-.8) {$\m{pick\ item}$\nodepart{two}$\{p_1\}$};
\node[trans, anchor=south,splitme] (pi1t) at (2,-.8) {$\m{pick\ item}$\nodepart{two}$\{p_1\}$};
\node[mixtrans, anchor=north,minimum width=13mm,splitme] (sm) at (3.8,-.8) {$\m{ship}$\nodepart{two}$\{o_1, p_2\}$};
\node[trans, anchor=south,minimum width=13mm,splitme] (st) at (3.8,-.8) {$\m{ship}$\nodepart{two}$\{o_1, p_2\}$};
\draw[arc] (pom.north east) -- (pi1m.north west);
\draw[arc] (pom.north east) -- (pm.north west);
\begin{scope}[yshift=-3.5cm]
\node[mixtrans, anchor=north,splitme] (pom2) {$\m{place\ order}$\nodepart{two}$\{o_2, p_2\}$};
\node[trans, anchor=south,splitme] (pot2) {$\m{place\ order}$\nodepart{two}$\{o_2, p_2\}$};
\node[otrans, anchor=north,splitme] (pm2) at (2,-.8) {$\m{payment}$\nodepart{two}$\{o_2\}$};
\node[trans, anchor=south,splitme] (pt2) at (2,-.8) {$\m{payment}$\nodepart{two}$\{o_2\}$};
\node[itrans, anchor=north,splitme] (pi1m2) at (2,.8) {$\m{pick\ item}$\nodepart{two}$\{p_2\}$};
\node[trans, anchor=south,splitme] (pi1t2) at (2,.8) {$\m{pick\ item}$\nodepart{two}$\{p_2\}$};
\node[mixtrans, anchor=north,minimum width=13mm,splitme] (sm2) at (3.8,.8) {$\m{ship}$\nodepart{two}$\{o_2, p_1\}$};
\node[trans, anchor=south,minimum width=13mm,splitme] (st2) at (3.8,.8) {$\m{ship}$\nodepart{two}$\{o_2, p_1\}$};
\draw[arc] (pom2.north east) -- (pi1m2.north west);
\draw[arc] (pom2.north east) -- (pm2.north west);
\end{scope}
\draw[arc] (pi1m.north east) -- (sm2.north west);
\draw[arc] (pm.north east) -- (sm.north west);
\draw[arc] (pi1m2.north east) -- (sm.north west);
\draw[arc] (pm2.north east) -- (sm2.north west);
\end{scope}
\node at (3.4,-3.5) {(a) Object-centric Petri net in the formalism of~\cite{LissAA23}. };
\node at (11,-3.5) {(b) Optimal alignment.};
\end{tikzpicture}}
%\vspace*{-7mm}
\caption{Order process and alignment.\label{fig:intro}}
\end{figure}
\end{example}
%\vspace{-2mm}

Our goal %of this paper
 is to remedy problems like the one just presented, introducing a comprehensive framework for alignment-based conformance checking of Petri net-based object-centric processes, where object identity and synchronization are fully accommodated.
 %taken into consideration. %More in detail, w
Precisely, we tackle the following research questions:
\begin{compactenum}
\item[(1)] How to obtain a lightweight formalism for Petri net-based object-centric processes, capable to track object identity and express synchronization? % constructs?
\item[(2)] How to formalise conformance checking for such a model?  
\item[(3)] Is it possible to obtain a feasible, algorithmic technique tackling this form of conformance checking?
\end{compactenum}

We answer all these questions affirmatively. Pragmatically, we start from the approach in \cite{LissAA23}, notably from the object-centric Petri nets used therein, which provide two essential features: the presence of multiple object types, and a special type of arc expressing that \emph{many} tokens flow at once through it (how many is decided at binding time). We infuse these nets with explicit object manipulation as originally introduced in \emph{Petri nets with identifiers} \cite{WerfRPM22}. The resulting, novel model of \emph{object-centric Petri nets with identifiers} (\mynets) combines %the main features of the two models: 
simultaneous operations over multiple objects with one-to-many relations, and subset synchronisation based on the identity of objects and their mutual relations. 

We then show how the notion of alignment originally introduced in \cite{LissAA23} can be suitably lifted to the much richer formalism of \mynets. We discuss relevant examples showing how non-conformance like the case of \exaref{intro} can be detected. A side-result of independent interest arises as a result of this discussion: even though, from the modelling point of view, \mynets are only capable of expressing subset synchronisation, if the model does not provide a way to progress objects that missed the synchronization point, then alignments enforce the exact synchronization semantics. 
% Notably, this provides the basis for further extending our approach to the full synchronous proclet model.

Finally, we show how the problem of computing optimal object-centric alignments can be cast as a Satisfiability/Optimization Modulo Theory problem, in the style of \cite{FelliGMRW23}, but with a much more complex logical encoding, due to object-centricity. We prove correctness of our encoding, implement it in a new tool \thetool, and experimentally validate the feasibility of our implementation.

The remainder of this paper is structured as follows:
We first recall relevant background about object-centric event logs (\secref{background}),
before we compare object-centric process models from the literature (\secref{related}). We then introduce our notion of object-centric Petri nets with identifiers, and its conformance checking task (\secref{nets}).  %considered in this paper (\secref{nets}).
Next, we present our encoding of %the conformance checking task
this task as an SMT problem (\secref{encoding}). We then sketch the implementation of our approach in the tool \thetool and present experiments (\secref{evaluation}), before we conclude (\secref{conclusion}).
%\vspace{-3mm}

\section{Object-Centric Event Logs}
\label{sec:background}
%\vspace{-2mm}
We recall basic notions on object-centric event logs \cite{AalstB20}.
Let $\otypes$ be a set of object types, $\activities$ a set of activities, and $\timestamps$ a set of timestamps %equipped 
with a total order $<$.
\begin{definition}
An  \emph{event log} is a tuple
$\eventlog = \tup{\events, \objects, \projact, \projobj,\projtime}$
where
\begin{compactitem}
\item $\events$ is a set of event identifiers,
\item $\objects$ is a set of object identifiers that are typed by a function $\otype\colon \objects \to \otypes$,
\item the functions $\projact\colon E \to \activities$, $\projobj\colon E \to \mc P(\objects)$, and $\projtime\colon E \to \timestamps$ associate
each event $e\in\events$ with an activity, a set of affected objects, and a timestamp, respectively, such that
for every $o\in \objects$ the timestamps $\projtime(e)$ of all events $e$ such that $o \in \projobj(e)$ are all different.%
\end{compactitem}
\end{definition}
Given an event log and an object $o \in \objects$, we write $\projtrace(o)$
for the tuple of events involving $o$, ordered by timestamps. Formally, $\projtrace(o) =\tup{e_1, \dots, e_n}$ such that $\{e_1, \dots, e_n\}$ is the set 
of events in $E$ with $o\in \projobj(e)$, and $\projtime(e_1) < \dots <\projtime(e_n)$
(by assumption these timestamps can be totally ordered.)
% \footnote{We use the assumption to avoid an addiitonal component $\projtrace$ as in ~\cite{AalstB20}.}

In examples, we often leave $\objects$ and $\activities$
implicit and present an event log $\eventlog$ as a set of tuples 
$\tup{e,\projact(e), \projobj(e), \projtime(e)}$ 
representing events. Timestamps are shown as natural numbers, and concrete event ids as $\eid{0}, \eid{1}, \dots$.
\begin{example}
\label{exa:log}
Let $\objects=\{o_1, o_2, o_3, p_1, \dots, p_4\}$ with
$\otype(o_i)=\mathit{order}$ and $\otype(p_j)$ $=\mathit{product}$ for all $i$, $j$.
The following is an %object-centric 
event log with $\events = \{\eid{0}, \eid{1}, \dots, \eid{9}\}$:
\begin{footnotesize}
\begin{align*}
&\tup{\eid{0}, \m{place\ order}, \{o_1,p_1\}, 1}, 
\tup{\eid{1}, \m{payment}, \{o_1\}, 3}, 
\tup{\eid{2}, \m{pick\ item}, \{o_1,p_1\}, 2}, \notag \\
&\tup{\eid{3}, \m{place\ order}, \{o_2,p_2\}, 3}, 
\tup{\eid{4}, \m{payment}, \{o_2\}, 4}, 
\tup{\eid{5}, \m{pick\ item}, \{o_2,p_2\}, 5},  \notag \\
&\tup{\eid{6}{,} \m{ship}{,} \{o_1,p_2\}{,} 6}, 
\tup{\eid{7}{,} \m{ship}{,} \{o_2,p_1\}{,} 9},
\tup{\eid{8}{,} \m{payment}{,} \{o_3\}{,} 2}, 
\tup{\eid{9}{,} \m{ship}{,} \{o_3, p_3, p_4\}{,} 5}
\end{align*}
\end{footnotesize}
For instance, we have $\projtrace(o_1) = \tup{\eid{0}, \eid{1}, \eid{2}, \eid{6}}$.
\end{example}

Given an event log $\eventlog = \tup{\events, \objects, \projact, \projobj,\projtime}$,
the \emph{object graph} $\GG_\eventlog$ of $\eventlog$ is the undirected graph with node set $\objects$, and an edge from $o$ to $o'$ if there is some event $e\in E$ such that $o\in \projobj(e)$ and
$o'\in \projobj(e)$. Thus, the object graph indicates which objects share events.
We next define a \emph{trace graph} as the equivalent of a linear trace in our setting.%
\footnote{Trace graphs are called \emph{process executions} in \cite{AalstB20}, we reserve this term for model runs.}

\begin{definition}
\label{def:trace:graph}
Let $\eventlog = \tup{\events, \objects, \projact, \projobj,\projtime}$ be an event log, and $X$ a connected component in $\GG_\eventlog$.
The \emph{trace graph} induced by $X$ is the directed graph $\tracenet_X = \tup{E_X, D_X}$ where
\begin{compactitem}
\item
the set of nodes $E_X$ is the set of all events $e\in E$ that involve objects in $X$, i.e., such that $X \cap \projobj(e) \neq \emptyset$, and 
\item the set of edges $D_X$ consists of all $\tup{e,e'}$ such that for some 
$o\in \projobj(e) \cap \projobj(e')$, it is $\projtrace(o)=\tup{e_1,\dots,e_n}$
and $e=e_i$, $e'=e_{i+1}$ for some $0\,{\leq}\,i\,{<}\,n$.%
\end{compactitem}
\end{definition}
% We omit the subscript of a trace graph $\tracenet_X$ if irrelevant or clear from the context.%

\begin{example}
\label{exa:trace graph}
The object graph for the log in \exaref{log} is shown in \figref{ograph} on the left. It has two connected components $X_1$ and $X_2$, whose trace graphs $T_{X_1}$ and $T_{X_2}$ are shown on the right.
\begin{figure}[t]
\begin{tikzpicture}
\tikzstyle{event}=[scale=.8, inner sep=2pt]
\node[scale=.8] at (.4, -1) {$X_1$};
\node[scale=.8] at (2.5, -1) {$X_2$}; 
\node[xshift=40mm, scale=.8] at (1.9, -1) {$T_{X_1}$}; 
\node[xshift=100mm, scale=.8] at (.7, -1) {$T_{X_2}$};
\begin{scope}[scale=.7, yshift=2mm]
\tikzstyle{object}=[scale=.8, inner sep=2pt]
\node[object] (o1) at (0,0) {$o_1$};
\node[object] (o2) at (1,0) {$o_2$};
\node[object] (p1) at (0,-1) {$p_1$};
\node[object] (p2) at (1,-1) {$p_2$};                        
\node[object] (o3) at (3,0) {$o_3$};
\node[object] (p3) at (3,-1) {$p_3$};
\node[object] (p4) at (4,-1) {$p_4$};
\draw (o1) -- (p1) -- (o2) -- (p2) -- (o1);
\draw (o3) -- (p3) -- (p4) -- (o3);
\end{scope}
\begin{scope}[xshift=40mm, scale=.6]
\begin{scope}[xscale=1.6, yscale=.8]
\node[event] (0) {$\eid{0}$};
\node[event] (1) at (1, 1){$\eid{1}$};
\node[event] (2) at (1, -1){$\eid{2}$};
\node[event] (3) at (4, 0){$\eid{3}$};
\node[event] (4) at (3, -1){$\eid{4}$};
\node[event] (5) at (3, 1){$\eid{5}$};
\node[event] (6) at (2, 1){$\eid{6}$};
\node[event] (7) at (2, -1){$\eid{7}$};
\draw[->] (2) -- (1);
\draw[->] (0) -- (2);
\draw[->] (3) -- (4);
\draw[->] (3) -- (5);
\draw[->] (1) -- (6);
\draw[->] (5) -- (6);
\draw[->] (2) -- (7);
\draw[->] (4) -- (5);
% \draw[->] (2) -- (6);
\draw[->] (5) -- (7);
\end{scope}
\end{scope}
\begin{scope}[xshift=100mm, scale=.6]
\begin{scope}[xscale=1.6, yscale=.9]
\node[event] (8) {$\eid{8}$};
\node[event] (9) at (1,0){$\eid{9}$};
\draw[->] (8) -- (9);
\end{scope}
\end{scope}
\end{tikzpicture}
\caption{Object graph for \exaref{log} and the respective trace graphs.\label{fig:ograph}}
\end{figure}
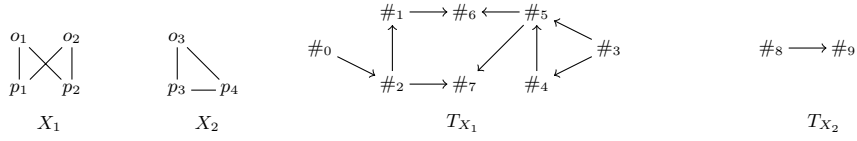
Note that $T_{X_1}$ corresponds to the trace in \exaref{intro}, just that \textsf{pick item} now mentions also the order it is associated with.
\end{example}

\section{Related Work and Modelling Features}
\label{sec:related}

To single out essential modelling features in the object-centric space, we %have
reviewed the literature, considering seminal/survey papers \cite{Aalst19,BeMA23,Aalst23} and papers proposing object-centric process
% modelling languages. 
models.
For the latter, we %limit ourselves 
restrict to 
those
%languages
that rely on a Petri net-based specification\:% of control flows
\cite{SoSD22,AalstB20,Fahland19,PWOB19,GhilardiGMR22,WerfRPM22}. A summary of the modell\-ing features described next and their support in these %different
approaches is given in \tabref{features}. 

A first essential feature is, as expected, the presence of constructs for \emph{creating and deleting objects}. Approaches distinguish each other depending on whether objects are \emph{explicitly referenced} in the model, or only \emph{implicitly manipulated}. Another essential feature is to allow objects to \emph{flow concurrently} and independently from each other --- e.g., items that are picked while their order is paid (cf.~\emph{divergence} in \cite{Aalst19}). Possibly, \emph{multiple objects} of the same type can be \emph{transferred} at once --- e.g., a single check over multiple items. At the same time, it is also crucial to capture single transitions that manipulate multiple objects of the same or different type at once (cf.~\emph{convergence} in \cite{Aalst19}). 
A first form of convergence is when a single transition takes a single, (parent) object and \emph{spawns unboundedly many (child) objects} of a different type, all related to that parent --- e.g., placing an order attaches unboundedly many items to it. If such a parent-children, \emph{one-to-many relation} is tracked, 
% other forms of convergence can be supported, this time dealing with synchronizing transitions.
also convergence in the form of synchronizing transitions can be supported.
Such transitions allow a parent object to evolve %from one state to another 
only if some or all its child objects are in a certain state; these two forms are resp. called \emph{subset} or \emph{exact synchronization}. Finally, advanced forms of \emph{coreference} can be used to inspect and evolve multiple related objects at once.

Resource-constraint $\nu$-Petri nets \cite{SoSD22} provide the first formalism dealing with a primitive form of object-centricity, where %. There, single
objects can be created and removed through the typical, explicit object-reference constructs of $\nu$-Petri nets \cite{RVFE10}. Objects can be related to resources (which are defined in the initial marking, and cannot be generated anew), but not to other objects. Alignment-based conformance checking is defined 
%there keeping 
taking
into account object identifiers and the resources they are related to, in an exact \cite{SoSD22} and approximate \cite{SoSD23} way.

Object-centric nets \cite{AalstB20} provide an implicit approach to object-centric processes. Places and transitions are assigned to different object types. Simple arcs match with a single object at once, while double arcs process arbitrarily many objects of a given type. However, as already indicated, object relations are not tracked, which prohibits % makes it impossible 
to capture object synchronization and coreference. Alignment-based conformance checking for this class of nets is studied in \cite{LissAA23}.

Synchronous proclets \cite{Fahland19} provide a formalism where objects and their mutual relations are tracked implicitly. Dedicated constructs are provided to deal with the different types of convergence described above, including subset and exact synchronization (but not supporting other forms of coreference). 
%For a given object type, standard Petri nets are used, hence 
Multi-object transfer is only approximately captured via iteration (picking objects one by one). Conformance checking has not yet been studied for proclets.

Variants of Petri nets with identifiers (PNIDs) are studied in \cite{PWOB19,GhilardiGMR22,WerfRPM22}, without dealing with conformance. PNIDs extend $\nu$-Petri nets and the model in \cite{SoSD22} by explicitly tracking objects and their relations, using tuples of identifiers. Differently from object-centric nets and proclets, no constructs are given to operate over unboundedly many objects with a single transition. As shown in \cite{GhilardiGMR22}, multi-object transfer and spawning, as well as subset synchronization, can be simulated using object coreference and iteration, while exact synchronization would require data-aware wholeplace operations, which are not supported.

\newcommand{\header}[1]{
  \multicolumn{1}{P{90}{1.8cm}}{#1}
}

\newcommand{\implicit}{\textsf{imp.}}
\newcommand{\explicit}{\textsf{exp.}}

\begin{table}[t]
\begin{footnotesize}
\begin{tabularx}{\textwidth}{
m{2.1cm}CCCCCCCCCCCCCCCCCCCCCc}
&
&
\header{object creation}
&
&
\header{object removal}
&
&
\header{concurrent object flows}
&
&
\header{multi-object transfer}
&
&
\header{multi-object spawning}
&
&
\header{object relations}
&
&
\header{subset\\ sync}
&
&
\header{exact\\ sync}
&
&
\header{coreference\\~}
&
&
\header{object reference}
&
&
\header{conformance}
\\
\toprule
resource-const. $\nu$-PNS
\cite{SoSD22}
&&
\cmark
&&
\cmark
&&
\cmark
&&
\xmark
&&
\xmark
&&
\xmark
&&
\xmark
&&
\xmark
&&
\xmark
&&
\explicit
&&
\cite{SoSD22,SoSD23}
\\
\midrule
object-centric nets \cite{AalstB20}
&&
\cmark
&&
\cmark
&&
\cmark
&&
\cmark
&&
\cmark
&&
\xmark
&&
\xmark
&&
\xmark
&&
\xmark
&&
\implicit
&&
\cite{LissAA23}
\\
\midrule
synchronous proclets \cite{Fahland19}
&&
\cmark
&&
\cmark
&&
\cmark
&&
\smark
&&
\cmark
&&
\cmark
&&
\cmark
&&
\cmark
&&
\xmark
&&
\implicit
&&
\xmark
\\
\midrule
PNID variants \cite{PWOB19,GhilardiGMR22,WerfRPM22}
&&
\cmark
&&
\cmark
&&
\cmark
&&
\smark
&&
\smark
&&
\cmark
&&
\smark
&&
\xmark
&&
\cmark
&&
\explicit
&&
\xmark
\\
\midrule
\textbf{\mynets}
&&
\cmark
&&
\cmark
&&
\cmark
&&
\cmark
&&
\cmark
&&
\cmark
&&
\cmark
&&
\smark
&&
\cmark
&&
\explicit
&&
\emph{here}
\\
\bottomrule
\end{tabularx}
\end{footnotesize}
\caption{Comparison of Petri net-based object-centric process modelling languages along main modelling features, tracking which approaches support conformance. {\cmark} indicates full, direct support, {\xmark} no support, and {\smark} indirect support.}
\label{tab:features}
\end{table}

\section{Object-Centric Petri Nets with Identifiers}
\label{sec:nets}

Using the literature analysis in \secref{related}, %provided before, 
we define \emph{object-centric Petri nets with identifiers} (\mynets), combining the features of PNIDs with those of object-centric nets. As in PNIDs, objects can be created in \mynets using $\nu$ variables, and tokens can carry objects or tuples of objects (accounting for object relations). In this way, objects can be associated with other objects, e.g., in situations as in \exaref{order}, a product can ``remember'' the order it belongs to. Arcs are labeled with (tuples of) variables to match with objects and relations. %New objects can be created using a similar mechanism as $\nu$-Petri nets~\cite{RVFE10}. 
%What PNIDs miss is the possibility of manipulating unboundedly many objects of a given type when firing a single transition, a feature that is at the core of object-centric nets. 
% We reconstruct this feature in \mynets,
Differently from PNIDs, \mynets also support multi-object spawning and transfer typical of object-centric nets,
by including special variables that match with \emph{sets of objects}. All in all, as shown in the last row of \tabref{features}, our new %this provides a 
formalism supports all features discussed in \secref{related}, with the exception of exact synchronization, for the same reason for which this is not supported in PNIDs. However, we put $\smark$ there, since as we will see below, this feature is implicitly supported when computing alignments. 
%\MM{\subsection{Modelling Features for Object-Centric Processes}}

%\MM{
%f
%}

\smallskip
\noindent
\textbf{Formal Definition.}
% To define \mynets, we need some additional notation.
First, we assume that every object type $\sigma \in \otypes$ has a domain $\Dom(\sigma) \subseteq \objects$,
given by all objects in $\objects$ of type $\sigma$. 
In addition to the types in $\Sigma$, we also consider list types with a base type in $\sigma$, denoted as $\listtype{\sigma}$.
As in colored Petri nets, each place has a \emph{color}: a cartesian product of data types from $\otypes$. 
More precisely, the set of colors $\colors$ is the set of all $\sigma_1 \times \cdots \times \sigma_m$ such that $m \geq 1$ and $\sigma_i \in \otypes$ for all $1\leq i\leq m$.
We fix a set of $\otypes$-typed variables $\allvars = \varset \uplus \varset_{list} \uplus \nuvarset$ as the disjoint union of 
a set $\varset$ of ``normal'' variables that refer to single objects, denoted by lower-case letters like $v$, with a type $\vartype(v) \in \Sigma$,
a set $\varset_{list}$ of list variables that refer to a list of objects of the same type, denoted by upper case letters like $U$, with a type $\vartype(U)=\listtype{\sigma}$ for some $\sigma \in \Sigma$,
and 
 a set $\nuvarset$ of variables referring to fresh objects, denoted as $\nu$, 
 with $\vartype(\nu) \in \Sigma$. 
We assume that infinitely many variables of each kind exist, and for every $\nu \,{\in}\,\nuvarset$, that $\Dom({\vartype(\nu)})$ is %countably 
infinite, to ensure an unbounded supply of fresh objects~\cite{RVFE10}.

In \mynets, tokens are tuples of objects, each associated with a color.
E.g., for the objects in \exaref{log}, we want to use $\tup{o_1, p_1}$, $\tup{o_2, p_2}$, or also just $\tup{o_1}$ as tokens.
To define relationships between objects in consumed and produced tokens when firing a transition,  we next define \emph{inscriptions} of arcs.

\begin{definition}
\label{def:inscription}
An \emph{inscription} is a tuple $\vec v = \tup{v_1, \dots, v_m}$
such that $m \geq 1$ and $v_i\in \allvars$ for all $i$, but at most one $v_i \in \listvarset$, for $1 \leq i \leq m$.
We call $\vec v$ a 
\emph{template inscription} if $v_i\in \listvarset$ for some $i$, and a
\emph{simple inscription} otherwise.
\end{definition}

For instance, for $o,p\in\varset$ and $P\in \listvarset$, there are inscriptions $\tup{o,P}$ or $\tup{p}$, the former being a template inscription and the latter a simple one. However, $\tup{P,P}$ is not a valid inscription as it has two list variables.
By allowing at most one list variable in inscriptions, we restrict to many-to-one relationships between objects, but it is known that many-to-many relationships can be modeled by many-to-one with auxiliary objects, through reification.
Template inscriptions will 
be used to capture an arbitrary number of tokens of the same color: intuitively,
if $o$ is of type \emph{order} and $P$ of type $\listtype{product}$, then
$\tup{o,P}$ refers to a single order with an arbitrary number of products.

We define the color of an inscription $\iota =\tup{v_1, \dots, v_m}$ as the tuple of the types of the involved variables, i.e., $\coloring(\iota) =\tup{\sigma_1, \dots, \sigma_m}$ where $\sigma_i=\vartype(v_i)$ if $v_i\in \varset\cup \nuvarset$, and $\sigma_i=\sigma'$ if $v_i$ is a list variable of type $\listtype{\sigma'}$.
Moreover, we set  $\vars{\iota} = \{v_1, \dots, v_m\}$.
E.g. for $\iota=\tup{o,P}$ with $o$, $P$ as above, we have $\coloring(\iota) =$ $\tup{\mathit{order}, \mathit{product}}$ and $\vars{\iota}=\{o,P\}$.
The set of all inscriptions is denoted $\Omega$.

\begin{definition}
\label{def:OCInet}
An \emph{object-centric Petri net with identifiers} (\mynet) %$\values$-typed
is defined as a tuple  
$N = (\otypes, \places, \transitions, \inflow, \outflow, 
% \inflowvar, \outflowvar, 
\coloring,\ell)$,
where:
\begin{compactenum}
\item $\places$ and $\transitions$ are finite sets of places and transitions such that $P\cap T=\emptyset$;
\item $\coloring\colon \places \rightarrow \colors$ maps every place to a color;
\item $\ell \colon T \to \activities \cup \{\tau\}$ is the transition labelling where $\tau$ marks an invisible activity,
\item $\inflow \colon \places \times \transitions \rightarrow \Omega$ is a partial function called \emph{input flow},  such that $\coloring(\inflow(p,t))=\coloring(p)$
for every $(p,t)\in \dom(\inflow)$;
\item $\outflow\colon \transitions \times \places \rightarrow \Omega$ is a partial function called \emph{output flow}, such that $\coloring(\outflow(t,p))=\coloring(p)$ for every $(t,p)\in \dom(\outflow)$;
\end{compactenum}
We set $\invars{t} = \bigcup_{p \in P} \vars{\inflow(p,t)}$, and
$\outvars{t} = \bigcup_{p \in P} \vars{\outflow(t,p)}$, and require that
$\invars{t}\,{\cap}\,\nuvarset\,{=}\,\emptyset$ and 
$\outvars{t} \subseteq \invars{t}\,{\cup}\,\nuvarset$, for all $t\,{\in}\,T$.
\end{definition}

If $\inflow(p, t)$ is defined, it
is called a \emph{variable} flow if $\inflow(p,t)$ is a template inscription,
and \emph{non-variable} flow otherwise; and similar for output flows.
Variable flows play the role of variable arcs in \cite{LissAA23}, they can carry multiple tokens at once.
For an \mynet $\NN$ as in \defref{OCInet}, we also use the common notations $\pre t = \{p \mid (p,t)\in \dom(\inflow)\}$ and $\post t = \{p \mid (t,p)\in \dom(\outflow)\}$.

\begin{example}
\label{exa:order}
\figref{OPID:package:handling} expresses the model of \exaref{intro} as an \mynet.
We use variables $\nu_o$ or type \emph{order} and $\nu_p$ of type \emph{product}, both in $\nuvarset$, to refer to new objects and products; as well as normal variables $o,p \in \varset$ of type \emph{order} and \emph{product}, respectively, and a variable $P$ of type $\listtype{\mathit{product}}$.
For readability, we write e.g. $o$ instead of $\tup{o}$.
Note that in contrast to the model in \figref{intro}, after \textsf{place\:order}, products now remember the order they belong to.
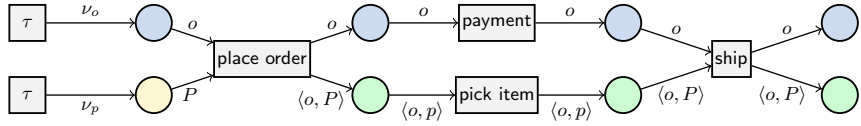
\begin{figure}[t]
\begin{tikzpicture}[node distance=26mm]
\node[oplace] (o0) {};
\node[trans, left of=o0, xshift=5mm] (geno) {$\tau$};
\node[oplace,right of=o0, xshift=10mm] (o1) {};
\node[trans,right of=o1, xshift=-5mm] (ot1) {$\m{payment}$};
\node[oplace,right of=ot1, xshift=-5mm] (o2) {};
\node[oplace,right of=o2, xshift=10mm] (o3) {};
\node[iplace, below of=o0, yshift=14mm] (i0) {};
\node[trans, left of=i0, xshift=5mm] (geni) {$\tau$};
\node[oiplace,right of=i0, xshift=10mm] (i1) {};
\node[trans,below of=o1,yshift=20mm,xshift=-18mm] (po) {$\m{place\ order}$};
\node[trans,right of=i1, xshift=-5mm] (it1) {$\m{pick\ item}$};
\node[oiplace,right of=it1, xshift=-5mm] (i2) {};
\node[oiplace,right of=i2, xshift=10mm] (i3) {} ;
\node[trans,below of=o3,yshift=20mm,xshift=-18mm] (ship) {$\m{ship}$};
\draw[arc] (geno) -- node[above, insc] {$\nu_o$} (o0);
\draw[arc] (geni) -- node[below, insc] {$\nu_p$} (i0);
\draw[arc] (o0) -- node[above, insc] {$o$} (po);
\draw[arc] (po) -- node[above, insc] {$o$} (o1);
\draw[arc] (o1) -- node[above, insc] {$o$} (ot1);
\draw[arc] (ot1) -- node[above, insc] {$o$} (o2);
\draw[arc] (o2) -- node[above, insc] {$o$} (ship);
\draw[arc] (ship) -- node[above, insc] {$o$} (o3);
\draw[arc] (i0) -- node[below, insc, anchor=east, yshift=-2mm, xshift=2mm] {$P$} (po);
\draw[arc] (po) -- node[below, insc, anchor=west, yshift=-3mm, xshift=-7mm] {$\tup{o,P}$} (i1);
\draw[arc] (i1) -- node[below, insc] {$\tup{o,p}$} (it1);
\draw[arc] (it1) -- node[below, insc] {$\tup{o,p}$} (i2);
\draw[arc] (i2) -- node[insc, below, anchor=east, yshift=-3mm, xshift=6mm] {$\tup{o,P}$} (ship);
\draw[arc] (ship) -- node[insc, below, anchor=west,yshift=-3mm,xshift=-6mm] {$\tup{o,P}$} (i3);
\end{tikzpicture}
\caption{\mynet for a package handling process.\label{fig:OPID:package:handling}}
\end{figure}
As in \exaref{intro}, the transitions \textsf{place\:order} and \textsf{ship} have variable arcs, to process multiple products at once.
However, we can now use tokens that combine a product $p$ with its order $o$ in a tuple $\tup{o,p}$. In \exaref{alignment} below we will see that the semantics of \mynets ensure that in the \textsf{ship} transition, all consumed tokens must refer to \emph{the same} order.
\end{example}

\smallskip
\noindent
\textbf{Semantics.}
Given the set of objects $\objects$,
the set of \emph{tokens} $\objtuples$ is the set of object tuples $\objtuples\,{=}\,\{\objects^m\,{\mid}\,m\,{\geq}\,1\}$.
The \emph{color} of a token $\omega\in \objtuples$ of the form $\omega=\tup{o_1, \dots, o_m}$ is denoted
$\coloring(\omega)=\tup{\vartype(o_1), \dots, \vartype(o_m)}$.
To define the execution semantics, we first introduce a notion of
a \emph{marking} of an \mynet $\NN=\tup{\otypes, \places, \transitions, \inflow, \outflow, \coloring,\ell}$, namely as a function $\marking\colon\places\rightarrow 2^{\objtuples}$, such that for all $p\in P$ and $\tup{o_1,\dots o_m} \in \marking(p)$, it holds that $\coloring(\tup{o_1,\dots o_m}) = \coloring(p)$. 
Let $\mathit{Lists}(\objects)$ denote the set of objects lists of the form $[o_1, \dots, o_k]$ with $o_1,\dots, o_k \in \objects$ such that all $o_i$ have the same type;
the type of such a list is then $\listtype{\vartype(o_1)}$.
Next, we define \emph{bindings} to fix which objects are involved in a transition firing.

\begin{definition}
A \emph{binding} for a transition $t$ and a marking $\marking$ is a type-preserving function 
$b\colon \invars{t} \cup \outvars{t} \to \objects \cup \mathit{Lists}(\objects)$.
To ensure freshness of created values, we demand that $b$ is injective on $\nuvarset \cap \outvars{t}$, and that $b(\nu)$ does not occur in $\marking$ for all $\nu \in \nuvarset \cap \outvars{t}$.
\end{definition}

For transition \textsf{ship} in \exaref{order} the mapping $b$ that sets
$b(o) = o_1$ and $b(P)=[p_1,p_2,p_3]$ is a binding (for any marking).
Next, we extend bindings to inscriptions to fix which tokens (not just single objects) participate in a transition firing.
The extension of a binding $b$ to inscriptions, i.e., variable tuples, is denoted $\vec b$.
For an inscription $\iota\,{=}\,\tup{v_1, \dots, v_m}$, let $o_i\,{=}\,b(v_i)$ for all $1\,{\leq}\,i\,{\leq}\,m$.
Then $\vec b(\iota)$ is the set of object tuples defined as follows:
if $\iota$ is a simple inscription then
$\vec b(\iota) = \{\tup{o_1, \dots, o_m}\}$.
Otherwise, there must be one $v_i$, $1\,{\leq}\,i\,{\leq}\,n$, such that $v_i \in \listvarset$, and consequently $o_i$ must be a list, say $o_i=[u_1, \dots, u_k]$ for some $u_1, \dots, u_k$.
Then $\vec b(\iota) = \{\tup{o_1, \dots, o_{i-1}, u_1, o_{i+1},\dots, o_m}, \dots, \tup{o_1, \dots, o_{i-1}, u_k, o_{i+1},\dots, o_m}\}$.
The set of all bindings is denoted by $\mathcal B$.
Next, we define that a transition with a binding $b$ is enabled if all object tuples pointed by $\vec b$ occur in the current marking. 

\begin{definition}
\label{def:enabled}
A transition $t \in\transitions$ and a binding $b$ for marking $\marking$ are \emph{enabled} in $\marking$
if $\vec b(\inflow(p,t)) \subseteq \marking(p)$ for all $p \in \pre{t}$.
\end{definition}
 
E. g., the binding $b$ with $b(o) = o_1$ and $b(P)=[p_1,p_2,p_3]$ is enabled in a marking $\marking$ of the \mynet in \exaref{order} with $\tup{o_1}\in \marking(q_{\mathit{blue}})$ and $\tup{o_1,p_1}, \tup{o_1,p_2}, \tup{o_1,p_3}\in \marking(q_{\mathit{green}})$, for $q_{\mathit{blue}}$ and $q_{\mathit{green}}$ the input places of \textsf{ship} with respective color.
 
\begin{definition}
Let transition $t$ be enabled in marking $\marking$ with binding $b$.
The \emph{firing} of $t$ yields the new marking $\marking'$ given by $\marking'(p)=\marking(p) \setminus \vec b(\inflow(p,t)) $ for all $p \in \pre t$, and
$\marking'(p)=\marking(p) \cup \vec b(\outflow(p,t))$ for all $p \in \post t$.
\end{definition}

We write $\marking \goto{t,b} \marking'$ to denote that $t$ is enabled with binding $b$ in $\marking$, and its firing yields $\marking'$.
A sequence of transitions with bindings 
$\run = \tup{(t_1, b_1), \dots, (t_n, b_n)}$ is called a \emph{run} 
if $\marking_{i-1} \goto{t_i, b_i} \marking_i$ for all $1\leq i \leq n$,
in which case we write $\marking_0 \goto{\run} \marking_n$.
For such a binding sequence $\run$, the \emph{visible subsequence} $\run_v$ is the subsequence of $\run$ consisting of all $(t_i, b_i)$ such that $\ell(t_i) \neq \tau$.

An \emph{accepting} object-centric Petri net with identifiers is an object-centric Petri net $\NN$
together with a set of initial markings $M_{\mathit{init}}$ and a set of final markings  $M_{\mathit{final}}$.
For instance, for \exaref{order}, $M_{\mathit{init}}$ consists only of the empty marking, whereas $M_{\mathit{final}}$ consists of all (infinitely many) markings in which each of the two right-most places has at least one token, and all other places have no token.
The \emph{language} of the net is given by
$\LL(\NN) = \{\run_v \mid m \goto{\run} m',\ m \in M_{\mathit{init}}\text{, and } m' \in M_{\mathit{final}}\}$,
i.e., the set of visible subsequences of accepted sequences.

\smallskip
\noindent
\textbf{Alignments.} In our approach, alignments show how a trace graph relates to a run of the model.
In the remainder of this section, we consider a trace graph $\tracenet_X = \tup{E_X, D_X}$ and an accepting \mynet with identifiers $\NN=\tup{\otypes, \places, \transitions, \inflow, \outflow, \coloring,\ell}$, and we assume that the language of $\NN$ is not empty.
% Let $\activities_\tau$ abbreviate $\activities \cup \{\tau\}$.

\begin{definition}
A \emph{move} is a tuple that is 
a \emph{model} move if it is in the set 
$\{\SKIP\} \times ((\activities \cup \{\tau\}) \times\mathcal P(\objects))$,
and a \emph{log move} if it is in the set
$(\activities \times\mathcal P(\objects)) \times \{\SKIP\}$, and 
a \emph{synchronous} move if it is of the form
$\tup{\tup{a,O},\tup{a',O'}} \in (\activities \times\mathcal P(\objects)) \times (\activities \times\mathcal P(\objects))$ such that $a=a'$ and $O=O'$.
The set of all synchronous, model, and log moves over $X$ and $\NN$ is denoted $\moves(\tracenet_X, \NN)$.
\end{definition}

\noindent In the object-centric setting, an alignment is a \emph{graph} of moves.
To define them formally, we first define log and model projections, similarly to~\cite{LissAA23}.
In a graph of moves we write $\tup{q_0,r_0} \skippath{\log} \tup{q_k,r_k}$  if there
is a path $\tup{q_0,r_0} \to \tup{\SKIP,r_1} \to \dots \tup{\SKIP,r_{k-1}} \to \tup{q_k,r_k}$, for $q_0, q_k \neq \SKIP$  and $k>0$, i.e., a path where all intermediate log components are $\SKIP$.
Similarly, $\tup{q_0,r_0} \skippath{\mod} \tup{q_k,r_k}$ abbreviates
 a path $\tup{q_0,r_0} \to \tup{q_1,\SKIP} \to \dots \tup{q_{k-1},\SKIP} \to \tup{q_k,r_k}$, for $r_0, r_k \neq \SKIP$ and $k>0$.

\begin{definition}[Projections]
Let $G{=}\tup{C,B}$ be a
graph with $C{\subseteq}\moves(\tracenet_X, \NN)$.

The \emph{log projection} $G\restrlog =\tup{C_l,B_l}$ is the graph with node set $C_l = \{ q \mid \tup{q,r} \in C\text{ and }q \neq \SKIP\}$, and an edge $\tup{q, q'}$ iff $\tup{q,r} \skippath{\log} \tup{q',r'}$ for some $r,r'$.

The \emph{model projection} $G\restrmod = \tup{C_m,B_m}$ is the graph with node set $C_m = \{ r \mid \tup{q,r} \in C\text{, }r \neq \SKIP\}$, and an edge $\tup{r, r'}$ iff $\tup{q,r} \skippath{\mod} \tup{q',r'}$ for some $q,q'$.
\end{definition}

Basically, for a graph $G$ over moves, the log projection is a graph that restricts to the log component of moves, omitting skip symbols.
The edges are as in $G$, except that one also adds edges that ``shortcut'' over model moves, i.e. where the log component is $\SKIP$; the model projection is analogous for the other component.
Next we define an alignment as a graph over moves where the log and model projections are a trace graph and a run, respectively.

\begin{definition}[Alignment]
An \emph{alignment} of a trace graph $\tracenet_X$ and an accepting \mynet $\NN$ is an acyclic directed graph $\Gamma=\tup{C,B}$ with $C \subseteq \moves(\tracenet_X, \NN)$
such that $\Gamma\restrlog=\tracenet_X$, 
there is a run $\run = \tup{\tup{t_1, b_1}, \dots, \tup{t_n,b_n}}$ with $\run_v \in \LL(\NN)$, and the model projection $\Gamma\restrmod = \tup{C_m,B_m}$ admits a bijection $f\colon \{\tup{t_1, b_1}, \dots, \tup{t_n,b_n}\} \to C_m$ such that
\begin{compactitem}
\item 
if $f(t_i, b_i) = \tup{a, O_m}$ then $\ell(t_i)=a$, and $O_m= \mathit{range}(b_i)$, for all $1\leq i \leq n$
\item
for all $\tup{r,r'}\in B_m$ there are $1{\leq}i{<}j{\leq}n$ such that $f(t_i,b_i){=}r$ and $f(t_j,b_j){=}r'$,
\SW{and conversely, if $r,r'\in C_m$ with $r=f(t_i,b_i)$ and $r'=f(t_{i+1},b_{i+1})$ for some $1\leq i < n$ then $\tup{r,r'} \in B_m$}.
\end{compactitem}
\end{definition}

\begin{example}
\label{exa:alignment}
\figref{alignment} shows an alignment $\Gamma$ for $T_{X_1}$ from \exaref{trace graph} w.r.t.~the model in \exaref{order}, where the log (resp. model) component is shown on top (resp. bottom) of moves.
Below, the log (left) and model projections (right) of $\Gamma$ are shown.
One can check that the former is isomorphic to $T_{X_1}$ in \exaref{trace graph}.
The alignment in \exaref{intro} is not valid in our setting: its model projection is not in the language of the net.
\begin{figure}[t]
\centering
\resizebox{.7\textwidth}{!}{
\begin{tikzpicture}[xscale=1.4, yscale=.7]
\begin{scope}
\node[trans, model, anchor=north,minimum width=11mm,splitme] (co1) at (-1.5,.8) {$\tau$\nodepart{two}$\{o_1\}$};
\node[trans, log, anchor=south,minimum width=11mm] (co1s) at (-1.5,.8) {$\SKIP$};
\node[trans, model, anchor=north,minimum width=11mm,splitme] (ci1) at (-1.5,-.8) {$\tau$\nodepart{two}$\{p_1\}$};
\node[trans, log, anchor=south,minimum width=11mm] (ci1s) at (-1.5,-.8) {$\SKIP$};
\node[trans, model, anchor=north,splitme] (pom) {$\m{place\ order}$\nodepart{two}$\{o_1,p_1\}$};
\node[trans, log, anchor=south,splitme] (pot) {$\m{place\ order}$\nodepart{two}$\{o_1,p_1\}$};
\node[trans, model, anchor=north,splitme] (pm) at (3,.8) {$\m{payment}$\nodepart{two}$\{o_1\}$};
\node[trans, log, anchor=south,splitme] (pt) at (3,.8) {$\m{payment}$\nodepart{two}$\{o_1\}$};
\node[trans, model, anchor=north,splitme] (pi1m) at (2,-.8) {$\m{pick\ item}$\nodepart{two}\parbox{12mm}{$\{o_1,p_1\}$}};
\node[trans, log, anchor=south,splitme] (pi1t) at (2,-.8) {$\m{pick\ item}$\nodepart{two}\parbox{12mm}{$\{o_1,p_1\}$}};
\node[trans, model, anchor=north,minimum width=20mm] (sm) at (5,.8) {$\SKIP$};
\node[trans, log, anchor=south,minimum width=20mm,splitme] (st) at (5,.8) {$\m{ship}$\nodepart{two}$\{o_1, p_2\}$};
\node[trans, model, anchor=north,minimum width=20mm,splitme] (sma) at (5,-.8) {$\m{ship}$\nodepart{two}$\{o_1, p_1\}$};
\node[trans, log, anchor=south,minimum width=20mm] (sta) at (5,-.8) {$\SKIP$};
\draw[arc] (pom.north east) -- (pi1m.north west);
\draw[arc] (pi1t.north) -- (pm.north west);
\draw[arc] (co1.north east) -- (pom.north west);
\draw[arc] (ci1.north east) -- (pom.north west);
\begin{scope}[yshift=-3.5cm]
\node[trans, model, anchor=north,minimum width=11mm,splitme] (co2) at (-1.5,.8) {$\tau$\nodepart{two}$\{o_2\}$};
\node[trans, log, anchor=south,minimum width=11mm] (co2s) at (-1.5,.8) {$\SKIP$};
\node[trans, model, anchor=north,minimum width=11mm,splitme] (ci2) at (-1.5,-.8) {$\tau$\nodepart{two}$\{p_2\}$};
\node[trans, log, anchor=south,minimum width=11mm] (ci2s) at (-1.5,-.8) {$\SKIP$};
\node[trans, model, anchor=north,splitme] (pom2) {$\m{place\ order}$\nodepart{two}$\{o_2, p_2\}$};
\node[trans, log, anchor=south,splitme] (pot2) {$\m{place\ order}$\nodepart{two}$\{o_2, p_2\}$};
\node[trans, model, anchor=north,splitme] (pm2) at (2,-.8) {$\m{payment}$\nodepart{two}$\{o_2\}$};
\node[trans, log, anchor=south,splitme] (pt2) at (2,-.8) {$\m{payment}$\nodepart{two}$\{o_2\}$};
\node[trans, model, anchor=north,splitme] (pi1m2) at (3,.8) {$\m{pick\ item}$\nodepart{two}\parbox{12mm}{$\{o_2,p_2\}$}};
\node[trans, log, anchor=south,splitme] (pi1t2) at (3,.8) {$\m{pick\ item}$\nodepart{two}\parbox{12mm}{$\{o_2,p_2\}$}};
\node[trans, model, anchor=north,minimum width=20mm] (sm2) at (5,.8) {$\SKIP$};
\node[trans, log, anchor=south,minimum width=20mm,splitme] (st2) at (5,.8) {$\m{ship}$\nodepart{two}$\{o_2, p_1\}$};
\node[trans, model, anchor=north,minimum width=20mm,splitme] (smb) at (5,-.8) {$\m{ship}$\nodepart{two}$\{o_2, p_2\}$};
\node[trans, log, anchor=south,minimum width=20mm] (stb) at (5,-.8) {$\SKIP$};
\draw[arc] (pom2.north east) -- (pi1m2.north west);
\draw[arc] (pom2.north east) -- (pm2.north west);
\draw[arc] (co2.north east) -- (pom2.north west);
\draw[arc] (ci2.north east) -- (pom2.north west);
\end{scope}
\draw[arc] (pi1m.north east) -- (sm2.north west);
\draw[arc] (pm.north east) -- (sm.north west);
\draw[arc] (pi1m2.north east) -- (sm.north west);
% \draw[arc] (pm2.north east) -- (sm2.north west);
\draw[arc] (pi1m.north east) -- (sma.north west);
\draw[arc] (pm.north east) -- (sma.north west);
\draw[arc] (pi1m2.north east) -- (smb.north west);
\draw[arc] (pi1m2.north east) -- (st2.south west);
\draw[arc] (pt2) -- (pi1m2);
\end{scope}
\end{tikzpicture}}
\\[1ex]
\resizebox{\textwidth}{!}{
\begin{tikzpicture}[xscale=1.4, yscale=.4] % log projection
\node[trans, log, anchor=south,splitme] (pot) {$\m{place\ order}$\nodepart{two}$\{o_1,p_1\}$};
\node[trans, log, anchor=south,splitme] (pt) at (2,.8) {$\m{payment}$\nodepart{two}$\{o_1\}$};
\node[trans, log, anchor=south,splitme] (pi1t) at (2,-.8) {$\m{pick\ item}$\nodepart{two}$\{o_1,p_1\}$};
\node[trans, log, anchor=south,minimum width=20mm,splitme] (st) at (3.8,0) {$\m{ship}$\nodepart{two}$\{o_1, p_2\}$};
\begin{scope}[yshift=-3.5cm]
\node[trans, log, anchor=south,splitme] (pot2) {$\m{place\ order}$\nodepart{two}$\{o_2, p_2\}$};
\node[trans, log, anchor=south,splitme] (pt2) at (2,-.8) {$\m{payment}$\nodepart{two}$\{o_2\}$};
\node[trans, log, anchor=south,splitme] (pi1t2) at (2,.8) {$\m{pick\ item}$\nodepart{two}$\{o_1,p_2\}$};
\node[trans, log, anchor=south,minimum width=20mm,splitme] (st2) at (3.8,0) {$\m{ship}$\nodepart{two}$\{o_2, p_1\}$};
\end{scope}
\draw[arc] (pot) -- (pi1t);
\draw[arc] (pi1t) -- (pt);
\draw[arc] (pt.east) -- (st.175);
\draw[arc] (pi1t2.east) -- (st.185);
\draw[arc] (pot2) -- (pt2);
\draw[arc] (pot2) -- (pi1t2);
\draw[arc] (pi1t.east) -- (st2.175);
\draw[arc] (pi1t2.east) -- (st2.185);
\draw[arc] (pt2) -- (pi1t2);
\begin{scope}[xshift=6.8cm,yshift=12mm] % model projection
\node[trans, model, anchor=north,minimum width=11mm,splitme] (co1) at (-1.5,.8) {$\tau$\nodepart{two}$\{o_1\}$};
\node[trans, model, anchor=north,minimum width=11mm,splitme] (ci1) at (-1.5,-.8) {$\tau$\nodepart{two}$\{p_1\}$};
\node[trans, model, anchor=north,splitme] (pom) {$\m{place\ order}$\nodepart{two}$\{o_1,p_1\}$};
\node[trans, model, anchor=north,splitme] (pm) at (2,.8) {$\m{payment}$\nodepart{two}$\{o_1\}$};
\node[trans, model, anchor=north,splitme] (pi1m) at (2,-.8) {$\m{pick\ item}$\nodepart{two}$\{o_1,p_1\}$};
\node[trans, model, anchor=north,minimum width=20mm,splitme] (sma) at (3.8,0) {$\m{ship}$\nodepart{two}$\{o_1, p_1\}$};
\begin{scope}[yshift=-3.5cm]
\node[trans, model, anchor=north,minimum width=11mm,splitme] (co2) at (-1.5,.8) {$\tau$\nodepart{two}$\{o_2\}$};
\node[trans, model, anchor=north,minimum width=11mm,splitme] (ci2) at (-1.5,-.8) {$\tau$\nodepart{two}$\{p_2\}$};
\node[trans, model, anchor=north,splitme] (pom2) {$\m{place\ order}$\nodepart{two}$\{o_2, p_2\}$};
\node[trans, model, anchor=north,splitme] (pm2) at (2,-.8) {$\m{payment}$\nodepart{two}$\{o_2\}$};
\node[trans, model, anchor=north,splitme] (pi1m2) at (2,.8) {$\m{pick\ item}$\nodepart{two}$\{o_2,p_2\}$};
\node[trans, model, anchor=north,minimum width=20mm,splitme] (smb) at (3.8,0) {$\m{ship}$\nodepart{two}$\{o_2, p_2\}$};
\end{scope}
\draw[arc] (pi1m) -- (pm);
\draw[arc] (pom) -- (pi1m);
\draw[arc] (pm.east) -- (sma);
\draw[arc] (pi1m.east) -- (sma);
\draw[arc] (pom2) -- (pm2);
\draw[arc] (pom2) -- (pi1m2);
\draw[arc] (pi1m2.east) -- (smb);
\draw[arc] (pm2) -- (pi1m2);
\draw[arc] (co1.east) -- (pom.175);
\draw[arc] (ci1.east) -- (pom.185);
\draw[arc] (co2.east) -- (pom2.175);
\draw[arc] (ci2.east) -- (pom2.185);
\end{scope}
\end{tikzpicture}}
\caption{Alignment with projections.\label{fig:alignment}}
\end{figure}
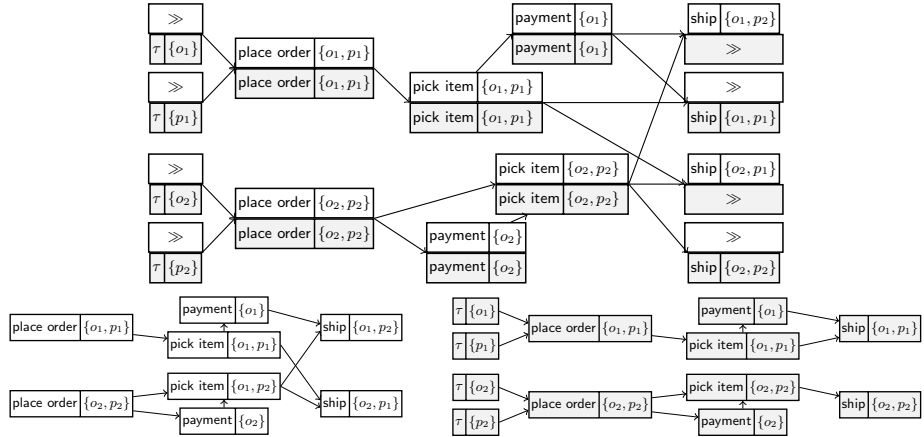
\end{example}

% In the sequel we use the following cost function based on \cite{LissAA23}.
In the remainder of the paper, we use the following cost function based on \cite{LissAA23}.
However, the approach developed below can also be adapted to other definitions.

\begin{definition}
\label{def:cost}
The cost of a move $M$ is defined as follows:
\begin{compactitem}
\item
if $M$ is a log move $\tup{\tup{a_\log, O_\log}, \SKIP}$ then $\cost(M) = |O_\log|$,
\item
if $M$ is a model move $\tup{\SKIP, \tup{a_\mod, O_\mod}}$ then $\cost(M) = 0$ if $a_\mod=\tau$, and $\cost(M) = |O_\mod|$ otherwise, 
\item
if $M$ is a synchronous move $\tup{\tup{a_\log, O_\log}, \tup{a_\mod, O_\mod}}$ then $\cost(M) = 0$.
\end{compactitem}
For an alignment $\Gamma=\tup{C,B}$, we set $\cost(\Gamma) =\sum_{M\in C}\cost(M)$, i.e., the cost of an alignment $\Gamma$ is  the sum of the cost of its moves.
\end{definition}

E.g., $\Gamma$ in \figref{alignment} has cost 8, as it involves two log moves and two non-silent model moves with two objects each. In fact, it is optimal in the following sense:

\begin{definition}
An alignment $\Gamma$ of a trace graph $\tracenet_X$ %in a log $L$ 
and an accept\-ing \mynet $\NN$ is \emph{optimal} 
if $\cost(\Gamma)\,{\leq}\,\cost(\Gamma')$ for all alignments $\Gamma'$ of $\tracenet_X$ and $\NN$.
\end{definition}

The \emph{conformance checking task} for an accepting \mynet $\NN$ and a log $L$ is
to find optimal alignments with respect to $\NN$ for all trace graphs in $L$.

\noindent
Finally, \mynets generalize the object-centric nets in~\cite{LissAA23,AalstB20}:

\begin{remark}
\label{rem:OPIfromtheirnet}
Every object-centric net $\mathit{ON}$ in the formalism of~\cite{LissAA23,AalstB20} can be encoded into an equivalent \mynet $\NN$ by adding suitable arc inscriptions, as follows. Let $v_\sigma\,{\in}\, \varset$ be a normal variable, and $V_\sigma\,{\in}\, \listvarset$ a list variable, for each object type $\sigma$ in $\mathit{ON}$.
For every arc $a$ that has as source or target a place $p$ taking objects of type $\sigma$,
we associate color $\tup{\sigma}$ with $p$, and add the arc inscription $\tup{V_\sigma}$ to $a$ if $a$ is a variable arc, and $\tup{v_\sigma}$ otherwise.
This is possible as in the object-centric Petri nets of~\cite{LissAA23,AalstB20}, transitions do not distinguish between objects of the same type.
Initial and final markings $M_{\mathit{init}}$ and $M_{\mathit{final}}$ for $\mathit{ON}$ can be used as such.
%
% Alternatively, for an \mynet where unboundedly many objects can be created as needed, $M_{\mathit{init}}$ can only contain the empty marking, provided that silent transitions with $\nu$-inscriptions are added to create fresh objects (cf. \exaref{order}).
\end{remark}

In \cite{LissAA23}, alignments are computed assuming that they involve exactly the set of objects mentioned in the log.
The next example shows that this assumption can compromise the existence of alignments, even irrespective of optimality.

\begin{example}
\label{exa:number:objects}
Consider the \mynet $\NN$ from \exaref{order}, and let $M_{\mathit{final}}$ consist of all markings where the right-most two places contain at least one token, and all other places contain no token. The empty run is thus not in the language of $\NN$.
Now consider the log
$L = \{\tup{\eid{0},\m{place\ order}, \{p\}, 1},\tup{\eid{1}.\m{pick\ item}, \{p\}, 2},$ $ \tup{\eid{2},\m{ship}, \{p\}, 3}\}$
 where (say, due to an error in the logging system), the order was not recorded, but only an item $p$ of type \emph{product}.
There exists no alignment involving only object $p$, as every run of $\NN$ must involve an object of type \emph{order}.
\end{example}

In contrast, an \mynet with silent transitions and $\nu$-inscriptions as in \exaref{order} to create objects, admits an optimal alignment also in this case.
We next observe that synchronization in \mynets is quite expressive when tackling conformance.

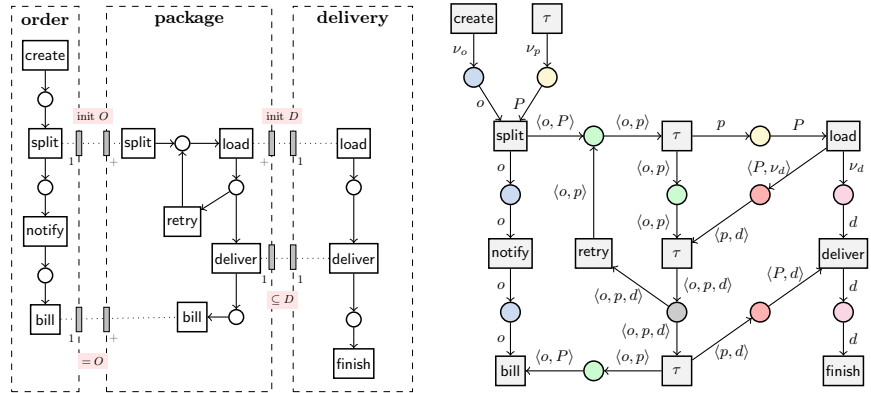
\begin{figure}[t]
\centering
\resizebox{.95\textwidth}{!}{
\begin{tikzpicture}[node distance=15mm]
\begin{scope}[]
\tikzstyle{placex} = [minimum width=4mm, inner sep=0pt]
\node[trans] (geno) {$\mathsf{create}$};
\node[trans, right of=geno] (genp) {$\tau$};
\node[oplace, placex, below of=geno, yshift=3mm] (o0) {};
\node[iplace, placex,below of=genp, yshift=3mm] (p0) {};
\node[trans, below of=p0, xshift=-7.5mm, yshift=3mm] (split) {$\mathsf{split}$};
\node[oiplace, placex,right of=split, xshift=2mm] (pack0) {};
\node[oplace, placex,below of=split, yshift=3mm] (order0) {};
\node[trans, below of=order0, yshift=3mm] (notify) {$\mathsf{notify}$};
\node[oplace, placex,below of=notify, yshift=3mm] (order1) {};
\node[trans, right of=pack0,xshift=2mm] (dummy) {$\tau$};
\node[iplace, placex,right of=dummy, xshift=2mm] (pack1) {};
\node[trans, right of=pack1, xshift=2mm] (load) {$\mathsf{load}$};
\node[oiplace, placex,below of=dummy, yshift=3mm] (pack2) {};
\node[idplace, placex, below of=pack1, yshift=3mm] (pack3) {};
\node[dplace, placex,below of=load, yshift=3mm] (del1) {};
\node[trans, below of=pack2, yshift=3mm] (join) {$\tau$}; 
\node[trans, right of=notify,xshift=2mm] (retry) {$\mathsf{retry}$};                
\node[oidplace, placex,below of=join, yshift=3mm] (pack4) {};
\node[trans, below of=pack4, yshift=3mm] (seal) {$\tau$};                
\node[idplace, placex,right of=pack4, xshift=2mm] (pack5) {};                 
\node[oiplace, placex,left of=seal, xshift=-2mm] (pack6) {}; 
\node[trans, below of=order1, yshift=3mm] (bill) {$\mathsf{bill}$}; 
\node[trans, below of=del1, yshift=3mm] (deliver) {$\mathsf{deliver}$}; 
\node[dplace, placex,below of=deliver, yshift=3mm] (del2) {};
\node[trans, below of=del2, yshift=3mm] (finish) {$\mathsf{finish}$}; 
\draw[arc] (geno) -- node[left, insc] {$\nu_o$} (o0);
\draw[arc] (genp) -- node[left, insc] {$\nu_p$} (p0);
\draw[arc] (o0) -- node[left, insc] {$o$} (split);
\draw[arc] (p0) -- node[left, insc] {$P$} (split);
\draw[arc] (split) -- node[above, insc] {$\tup{o,P}$} (pack0);
\draw[arc] (pack0) -- node[above, insc] {$\tup{o,p}$} (dummy);
\draw[arc] (dummy) -- node[above, insc] {$p$} (pack1);
\draw[arc] (dummy) -- node[left, insc] {$\tup{o,p}$} (pack2);
\draw[arc] (pack1) -- node[above, insc] {$P$} (load);
\draw[arc] (load) -- node[left, insc] {$\tup{P,\nu_d}$} (pack3); 
\draw[arc] (load) -- node[right, insc] {$\nu_d$} (del1); 
\draw[arc] (pack3) -- node[below, insc, xshift=2mm] {$\tup{p,d}$} (join);  
\draw[arc] (pack2) -- node[left, insc] {$\tup{o,p}$} (join);  
\draw[arc] (join) -- node[right, insc] {$\tup{o,p,d}$} (pack4); 
\draw[arc] (pack4) -- node[left, insc, near start, xshift=-1mm] {$\tup{o,p,d}$} (retry); 
\draw[arc] (retry) -- node[left, insc] {$\tup{o,p}$} (pack0);
\draw[arc] (pack4) -- node[left, insc, near start] {$\tup{o,p,d}$} (seal);  
\draw[arc] (seal) -- node[below, insc, xshift=2mm] {$\tup{p,d}$} (pack5);   
\draw[arc] (seal) -- node[above, insc] {$\tup{o,p}$} (pack6);   
\draw[arc] (pack6) -- node[above, insc] {$\tup{o,P}$} (bill); 
\draw[arc] (pack5) -- node[above, insc, xshift=-2mm] {$\tup{P,d}$} (deliver);  
\draw[arc] (del1) -- node[right, insc] {$d$} (deliver);  
\draw[arc] (deliver) -- node[right, insc] {$d$} (del2); 
\draw[arc] (del2) -- node[right, insc] {$d$} (finish); 
\draw[arc] (split) -- node[left, insc] {$o$} (order0);   
\draw[arc] (order0) -- node[left, insc] {$o$} (notify);  
\draw[arc] (notify) -- node[left, insc] {$o$} (order1);   
\draw[arc] (order1) -- node[left, insc] {$o$} (bill);
\end{scope}
% \draw[dashed, gray!50] (1.75,0) -- (1.75,-9);
% \draw[dashed, gray!50] (6,0) -- (6,-9);
\begin{scope}[xshift=-70mm, node distance=9mm, yshift=-6mm]
\tikzstyle{proctrans} = [trans, fill=white]
\tikzstyle{interface} = [draw, rectangle, fill=gray!50, minimum width=.9mm, minimum height=4mm, inner sep=0pt]
\tikzstyle{multiplicity} = [scale=.6, yshift=-5mm]
\tikzstyle{sync} = [scale=.6, fill=red!10, rectangle, inner sep = 3pt]
\tikzstyle{procplace} = [place, minimum width=3mm, inner sep=0pt]
\node[proctrans] (create) {$\mathsf{create}$};
\node[procplace, below of=create] (created) {};
\node[proctrans, below of=created] (split) {$\mathsf{split}$};
\node[procplace, below of=split] (sent) {};
\node[proctrans, below of=sent] (notify) {$\mathsf{notify}$};
\node[procplace, below of=notify] (notified) {};
\node[proctrans, below of=notified] (bill) {$\mathsf{bill}$}; 
\draw[arc] (create) -- (created);
\draw[arc] (created) -- (split);
\draw[arc] (split) -- (sent);
\draw[arc] (sent) -- (notify);
\draw[arc] (notify) -- (notified);
\draw[arc] (notified) -- (bill);
\draw[dashed] (-.55,.8) rectangle (.55, -5.5);
\node[scale=.9] at (0,.6) {\textbf{order}};
\node[scale=.9] at (2.35,.6) {\textbf{package}};
\node[proctrans, right of=split, xshift=10mm] (splitp) {$\mathsf{split}$};
\node[procplace, right of=splitp] (ready) {};
\node[proctrans, right of=ready, xshift=2mm] (load) {$\mathsf{load}$};
\node[procplace, below of=load] (loaded) {};
\node[proctrans, below of=ready, yshift=-7mm] (retry) {$\mathsf{retry}$};
\node[proctrans, below of=loaded, yshift=-5.5mm] (deliver) {$\mathsf{deliver}$};
\node[procplace, below of=deliver, yshift=-3mm] (done) {};
\node[proctrans, left of=done] (billp) {$\mathsf{bill}$};
\draw[arc] (splitp) -- (ready);
\draw[arc] (ready) -- (load);
\draw[arc] (load) -- (loaded);
\draw[arc] (loaded) -- (retry);
\draw[arc] (retry) -- (ready);
\draw[arc] (loaded) -- (deliver);
\draw[arc] (deliver) -- (done);
\draw[arc] (done) -- (billp);
\draw[dashed] (1,.8) rectangle (3.7, -5.5);
\draw[dotted] (split) -- (splitp);
\draw[dotted] (bill) -- (billp);
\node[interface, xshift=5.5mm] (is1) at (split) {}
 node[multiplicity, xshift=-2mm] at (is1) {1};
\node[interface, xshift=5.5mm] (ib1) at (bill) {}
 node[multiplicity, xshift=-2mm] at (ib1) {1};
\node[interface, xshift=10mm] (is2) at (split) {}
 node[multiplicity, xshift=2mm] at (is2) {+};
\node[interface, xshift=10mm] (ib2) at (bill) {}
 node[multiplicity, xshift=2mm] at (ib2) {+};
\node[proctrans, right of=load, xshift=15mm] (loadd) {$\mathsf{load}$};
\node[procplace, below of=loadd] (delivering) {};
% \node[proctrans, below of=delivering, xshift=-6mm, yshift=2mm] (retryd) {$\mathsf{retry}$};
\node[proctrans, below of=delivering, yshift=-5.5mm] (deliverd) {$\mathsf{deliver}$};
% \node[proctrans, below of=delivering, yshift=-13mm, xshift=6mm] (next) {$\mathsf{next}$};
\node[procplace, below of=delivering, yshift=-18mm] (delivered) {};
\node[proctrans, below of=delivered] (finish) {$\mathsf{finish}$};
\draw[arc] (loadd) -- (delivering);
% \draw[arc, rounded corners] (delivering) -| (retryd);
\draw[arc] (delivering) -- (deliverd);
\draw[arc] (deliverd) -- (delivered);
\draw[arc] (delivered) -- (finish);
% \draw[arc, rounded corners] (delivered) -| (next);
% \draw[arc, rounded corners] (retryd) |- (delivered);
% \draw[arc, rounded corners] (next) |- (delivering);
\draw[dashed] (4.05,.8) rectangle (6, -5.5);
\node[scale=.9] at (5.025,.6) {\textbf{delivery}};
\draw[dotted] (load) -- (loadd);
% \draw[dotted] (retry) -- (retryd);
\draw[dotted] (deliver) -- (deliverd);
\node[interface, xshift=37mm] (il1) at (split) {}
 node[multiplicity, xshift=-2mm] at (il1) {+};
\node[interface, xshift=40.5mm] (il2) at (split) {}
 node[multiplicity, xshift=2mm] at (il2) {1};
% \node[interface] (ir1) at (3.7, -2.7) {}
%  node[multiplicity, xshift=-2mm, yshift=1mm] at (ir1) {1};
% \node[interface] (ir2) at (4.05, -2.7) {}
%  node[multiplicity, xshift=2mm] at (ir2) {1};
\node[interface] (id1) at (3.7, -3.3) {}
 node[multiplicity, xshift=-2mm, yshift=-1mm] at (id1) {1};
\node[interface] (id2) at (4.05, -3.3) {}
 node[multiplicity, xshift=2mm, yshift=-1mm] at (id2) {1};
\node[sync] at (0.775, -1) {init $O$};
\node[sync] at (0.775, -5) {$=O$};
\node[sync] at (3.875, -1) {init $D$};
% \node[sync] at (3.875, -2.2) {$ \subseteq D$};
\node[sync] at (3.875, -4) {$ \subseteq D$};
\end{scope}
\end{tikzpicture}
}
%\vspace*{-4mm}
\caption{A synchronous proclet and a corresponding \mynet.}
\label{fig:proclet}
\end{figure}

\begin{example}
To show the expressiveness of \mynets, we model a proclet inspired from the running example in \cite{Fahland19}---keeping all essential proclet features. % this exemplifies the power of \mynets to capture the essential modeling features of proclets. %can also closely model proclets with synchronization:
The proclet $N$ is shown in Fig.~\ref{fig:proclet} on the left, while the \mynet on the right emulates the behaviour of $N$ with few differences:
Packages are created by a silent transition instead within $\mathsf{split}$,
and the exact synchronization ${=}O$ in the $\mathsf{bill}$ transition of $N$ cannot be expressed in \mynets. Instead, all synchronizations are implicitly of subset type.
Let a trace graph $T$ consist of the events
$\tup{\m{create}, \{o\},1}$, 
$\tup{\m{split}, \{o,p\},2}$, 
$\tup{\m{notify}, \{o\},3}$,
$\tup{\m{load}, \{d,p\},4}$, 
$\tup{\m{bill}, \{o\},5}$, 
$\tup{\m{deliver}, \{o, p\},5}$ and
$\tup{\m{finish}, \{d\},6}$.
Product $p$ was ignored in the $\m{bill}$ event, so the exact synchronization demanded by the proclet is violated.
In our \mynet, due to the implicit subset semantics, 
$T$ could be matched against a valid transition sequence, but no final marking would be reached: token $\tup{o,p}$ is left behind in the bottom green place.
Thus no alignment with a synchronous $\m{bill}$ move exists, hence no cost 0 alignment, and the lack of synchronization is in fact detected in the alignment task.
\end{example}

\section{Encoding}
\label{sec:encoding}
%\vspace{-1mm}

In encoding-based conformance checking, it is essential to fix upfront an upper bound on the size of an optimal alignment~\cite{FelliGMRW23}. For \mynets, the next lemma
establishes such a bound. %, though tighter bounds might be possible.
We assume that in runs of $\NN$ every object used in a silent transition occurs also in a non-silent one.
The precise statement with a formal proof can be found in~\cite[Lem.~2]{long-version}.

%\vspace{-1mm}
\begin{lemma}
\label{lem:bounds}
% Let  $\NN$ be \anet and $T_X=\tup{E_X,D_X}$ a trace graph with optimal alignment $\Gamma$.
% Let $m=\sum_{e\in E_X} |\projobj(e)|$ the number of object occurrences in $E_X$, and 
% $c = \sum_{i=1}^n |\dom(b_i)|$ the number of object occurrences in some run  $\rho$ of $\NN$, with $\rho_v=\tup{\tup{t_1,b_1}, \dots \tup{t_n,b_n}}$.
% Then $\Gamma|_\mod$ has 
% at most $(|E_X|+c+m)(k+1)$ moves if $\NN$ has no $\nu$-inscriptions, and at most $(|E_X|+3c+2m)(k+1)$ otherwise, where $k$ is the longest sequence of silent transitions without $\nu$-inscriptions in $\NN$.
% Moreover, $\Gamma|_\mod$ has at most $2c+m$ object occurrences in non-silent transitions.
%
Let $m$ be the number of object occurrences in $T_X$, $c$ the number of objects in some run of $\NN$, and $k$ the maximal number of subsequent silent transitions without $\nu$ in $\NN$.
The number of moves in $\Gamma|_\mod$ is linear in $|E_X|$, $c$, $m$ and $k$, and $\Gamma|_\mod$ has at most $2c{+}m$ object occurrences in non-silent transitions.
\end{lemma}

%\vspace{-2mm}
For our encoding we also need to determine a priori a set of objects $O$ 
that is a superset of the objects used in the model projection of the optimal alignment.
However, we can nevertheless solve problems as in \exaref{number:objects} using the bound on the number of objects  in \lemref{bounds}: e.g. in \exaref{number:objects}, we have $c=3$ and $m=6$, so the bound on object occurrences is 12, while 
the bound on moves is 25.
\newcommand{\transvar}{\mathtt T}
\newcommand{\markvar}{\mathtt M}
\newcommand{\objvar}{\mathtt O}
\newcommand{\distvar}{\mathtt d}
\newcommand{\rlvar}{\mathtt {len}}
\smallskip
\noindent\textbf{Encoding the run.}
Let $\NN = \tup{\otypes, \places, \transitions, \inflow, \outflow, \coloring,\ell}$ \anet and $T_X = \tup{E_X,D_X}$ a trace graph.
In the encoding, we assume the following, using \lemref{bounds}:
\begin{inparaenum}[\itshape (i)]
\item 
The number of nodes in the model projection of an optimal alignment is upper-bounded by a number $n\in \mathbb N$.
\item 
% the number of objects involved in any transition $t$ in an optimal alignment can be bounded by some $K\in \mathbb N$,
% so that we assume 
The set $O \subseteq \objects$ is a finite set of objects that might occur in the alignment, it must contain all objects in $T_X$, optionally it can contain more objects.
We assume that every $o\in O$ is assigned a unique id $id(o)\in \{1, \dots ,|O|\}$
% This bound includes objects that come from the instantiation of $\nu$ variables.
\item $K$ is the maximum number of objects involved in any transition firing, it can be computed from $O$ and $\NN$.
\end{inparaenum}

% We also assume that $\NN$ has a silent \emph{wait} transition in every final marking. 
% While this is not strictly necessary, it simplifies the encoding, as also noted in~\cite{BoltenhagenCC21,FelliGMRW23}.
\smallskip
\noindent
To encode a run $\rho$ of length at most $n$, we use the following SMT variables:
\begin{inparaenum}[(a)]
\item Transition variables $\transvar_j$ of type integer for all $1\leq j\leq n$ to identify 
the $j$-th transition in the run. To this end, we enumerate the transitions as $T = \{t_1, \dots, t_L\}$, and add the constraint $\bigwedge_{j=1}^n 1\leq \transvar_j \leq L$,
with the semantics that $\transvar_j$ is assigned value $l$ iff the $j$-th transition in $\rho$ is $t_l$.

\noindent\item To identify the markings in the run, we use marking variables $\markvar_{j,p,\vec o}$ of type boolean for every time point $0\leq j\leq n$, every place $p\in P$, and every vector $\vec o$ of objects with elements in $O$ such that $\coloring(\vec o)=\coloring(p)$. The semantics is that $\markvar_{j,p,\vec o}$ is assigned true iff $\vec o$ occurs in $p$ at time $j$.

\noindent\item To keep track of which objects are used by transitions of the run, we use object variables $\objvar_{j,k}$ of type integer for all $1\leq j\leq n$ and $0\leq k \leq K$ with the constraint $\bigwedge_{j=1}^n 1\leq \objvar_{j,k} \leq |O|$. 
The semantics is that if $\objvar_{j,k}$ is assigned value $i$ then, if $i>0$ the $k$-th object involved in the $j$-th transition is $o_i$, and if $i=0$ then the $j$-th transition uses less than $k$ objects.

\noindent\item To encode the actual length of the run, we use an integer variable $\rlvar$.
\end{inparaenum}

In addition, we use the following variables to represent alignment cost:
\begin{inparaenum}[(e)]
\item Distance variables $\distvar_{i,j}$ of type integer for every $0\leq i\leq m$ and $0\leq j\leq n$, their use will be explained later.
\end{inparaenum}
\smallskip

\noindent
Next we intuitively explain the used constraints: the formal description with all technical formulae of the encoding, as well as the decoding, can be found in~\cite{long-version}.
%While we explain the purpose of all encoding parts, we omit the more technical formulas here;
% the interested reader can find them 
%they can be found in~\cite{long-version}.
\begin{inparaenum}[(1)]

\noindent\item 
\emph{Initial markings}.
We ensure that the first marking in the run $\rho$ is initial.
%For a marking $M$, by the expression $[\vec o \in M(p)]$ we abbreviate $\top$ if an object tuple $\vec o$ occurs in $M(p)$, and $\bot$ otherwise.
%\vspace{-1mm}
%\begin{align}
%\notag
%\textstyle
%\bigvee_{M \in M_{init}}
%\bigwedge_{p\in P} \bigwedge_{\vec o \in \vec O_{\coloring(p)}} \markvar_{0,p, \vec o} = [\vec o \in M(p)] 
%\vspace{-1mm}
%\end{align}
%\vspace{-1mm}

\noindent\item 
%\vspace{-3mm}
\emph{Final markings.} 
Next, we state that after at most $n$ steps, but possibly earlier, a final marking is reached.
%\vspace{-2mm}
%\begin{align}
%\notag
%\textstyle
%\bigvee_{0 \leq j \leq n}
%\rlvar = j \wedge
%\bigvee_{M \in M_{final}}
%\textstyle\bigwedge_{p\in P} \bigwedge_{\vec o \in \vec O_{\coloring(p)}} \markvar_{j,p, \vec o} = [\vec o \in M(p)] 
%\end{align}
% \item 
% \vspace{-2mm}

\noindent\item 
\emph{Moving tokens.}
Transitions must be enabled; tokens are moved by transitions. 

\noindent\item 
\emph{Tokens not moved by transitions stay in their place.}
Similarly to the previous item, we capture that for every time point, place $p$, and token $\vec o$, the marking does not change for $p$ and $\vec o$ unless it is produced or consumed by some transition.
% \begin{align*}
% \bigwedge_{j=1}^{n} 
% \bigwedge_{p \in P}
% \bigwedge_{\vec o\in \vec O_{\coloring(p)}} (\markvar_{j-1,p,\vec o} \leftrightarrow \markvar_{j,p,\vec o}) \vee &\bigvee_{t_l \in \post{p}} (\transvar_j=l \wedge \constoken(p,t,j,\vec o) ) \vee {}\\
% & \bigvee_{t_l \in \pre{p}} (\transvar_j=l \wedge \prodtoken(p,t,j,\vec o) ) 
% \end{align*}

\noindent \item 
\emph{Transitions use objects of suitable type.}
To this end, recall that every transition can use at most $K$ objects, which limits instantiations of template inscriptions.
For every transition $t\in T$, we can thus enumerate the objects used by it from 1 to $K$.
Depending on the transition $t_l$ performed in the $j$-th step, we thus demand that the object variables $\objvar_{j,k}$ are instantiated by an object of suitable type, using a disjunction over all possible objects.
% 
% However, some of these objects may be unused. We use the shorthand $\mathit{needed}_{t,k}$ to express this: $\mathit{needed}_{t,k} = \top$ if the $k$-th object is necessary for transition $t$ because it occurs in a simple inscription, and $\bot$ otherwise.
% Moreover, let $\mathit{ttype}(t,k)$ be the type of the $k$-th object used by transition $t$.
% Finally, we denote by $O_\sigma$ the subset of objects in $O$ of type $\sigma$.
% \[
% \bigwedge_{j=1}^n \bigwedge_{l=1}^L \transvar_j=l \rightarrow \bigwedge_{k=1}^K \left ((\neg [\mathit{needed}_{t_l,k}] \wedge \objvar_{j,k}=0) \vee \bigvee_{o \in O_{\mathit{ttype}(t_l,k)}} \objvar_{j,k}=id(o)\right)
% \]

\noindent \item 
\emph{Objects that instantiate $\nu$-variables are fresh.}
Finally, we need to require that if an object is instantiated for a $\nu$-inscription in a transition firing, then this object does not occur in the current marking.
% 
% We assume in the following constraint that $tids_\nu$ is the set of all $1 \leq l \leq L$ such that $t_l$ has an outgoing $\nu$-inscription, and that every such $t_l$ has only one outgoing $\nu$-inscription $\nu_t$, and we assume w.l.o.g. that in the enumeration of objects of $t$, $\nu_t$ is the first object. However, the constraint can be easily generalized to more such inscriptions.
% \[
% \bigwedge_{j=1}^n \bigwedge_{l\in tids_\nu} \bigwedge_{o\in O_{\otype(\nu_t)}} \transvar_j=l \wedge \objvar_{j,1} = id(o)  \rightarrow 
% (\bigwedge_{p\in P} \bigwedge_{\vec o \in \vec O_{\coloring(p)}, o \in \vec o} \neg \markvar_{j-1,p,\vec o})
% \]
\end{inparaenum}

We denote by $\varphi_{run}$ the conjunction of the constraints in (a)--(c) and (1)--(6), as they encode a valid run of $\NN$.

\smallskip
\noindent\textbf{Encoding alignment cost.}
Similar as in~\cite{FelliGMRW23,BoltenhagenCC21}, we encode the cost of an alignment as the edit distance with respect to suitable penalty functions.
% using penalty functions $P_=$, $P_M$, and $P_L$. We illustrate here the formulae obtained for the standard cost function in Remark~ \ref{rem:distances}.
Given a trace graph
$T_X = (E_X, D_X)$, let
$e_1, \dots, e_m$ be an enumeration of all events in $E_X$ such that
$\projtime(e_1) \leq \dots \leq \projtime(e_m)$.
Let $[P_L]_i$, $[P_M]_{j}$, and $[P_=]_{i,j}$ be the penalty expressions, for a log move, a model move and a synchronous move, resp.
% , whose formal definition is in~\cite{long-version}.
%Let the penalty expressions $[P_L]_i$, $[P_M]_{j}$, and $[P_=]_{i,j}$
%be as follows, for all $1\leq i \leq m$ and $1 \leq j \leq n$:
%\begin{align*}
%[P_L]_{i} &= |\projobj(e_i)| \qquad\qquad [P_=]_{i,j} = 
%ite(\mathit{is\_labelled}(j, \projact(e_i)), 0, \infty)\\
%[P_M]_{j} &=
%ite(\mathit{is\_labelled}(j, \tau)\vee \rlvar <j, 0, \Sigma_{o\in O}\Sigma_{k=1}^K
%ite(\objvar_{j,k} = id(o), 1, 0))
%\end{align*}
%where $\mathit{is\_labelled}(j,a)$ expresses that the $j$-the transition has label $a\in \activities \cup \{\tau\}$, which can be done by taking $is\_labelled(j,a):= \bigvee_{l\in T_{idx}(a)}\transvar_j = l$
%where $T_{idx}(a)$ is the set of transition indices with label $a$, i.e., the set of all $l$ with $t_l \in T$ such that $\ell(t_l) = a$.
One can then encode the edit distance as in \cite{FelliGMRW23,BoltenhagenCC21}:
\begin{equation*}
\begin{array}{rl@{\qquad}rl@{\qquad}rl@{\qquad\quad}r}
\distvar_{0,0} &= 0 &
\distvar_{{i+1},0} &= [P_L] + \distvar_{i,0} &
\distvar_{0,{j+1}} &= [P_M]_{j+1} + \distvar_{0,j} 
&\hfill(\varphi_\delta)\\[1ex] % cheat label
\distvar_{i+1,j+1} &\multicolumn{5}{l}{=
\min (
[P_=]_{i+1, j+1} + \distvar_{i,j},\ 
[P_L] + \distvar_{i,j+1},\ 
[P_M]_{j+1} + \distvar_{i+1,j})}
\end{array}
\end{equation*}

\smallskip
\noindent\textbf{Solving.}
We use an SMT solver to obtain a satisfying assignment $\alpha$ for the 
following constrained optimization problem $(\Phi)$: \emph{$\varphi_{\mathit{run}} \wedge
\varphi_{\delta}
\text{\quad minimizing\quad }\distvar_{m,n}$}.
%
%\vspace{-2mm}
%\begin{align*}
%\label{eq:constraints}
%\varphi_{\mathit{run}} \wedge
%\varphi_{\delta}
%\text{\quad minimizing\quad }\distvar_{m,n}
%\tag{$\Phi$}
%\vspace{-3mm}
%\end{align*}
%\vspace{-2mm}
From $\alpha$ one can then decode the run $\rho(\alpha)$ and the optimal alignment $\Gamma(\alpha)$, similar as in~\cite{FelliGMRW23}. In particular, $\alpha(\distvar_{m,n})$ is the cost of the optimal alignment. 
% Details on the decoding can be found in~\cite{long-version}.

\begin{theorem}
Given \anet $\NN$, trace graph $\tracenet_X$, and satisfying assignment $\alpha$ to $(\Phi)$, $\Gamma(\alpha)$ is an optimal alignment of $\tracenet_X$ with respect to $\NN$ with cost $\alpha(\distvar_{m,n})$.
\end{theorem}

%\vspace{-3mm}

% 
% \noindent
% \textbf{(4) Decoding.}
% We obtain a valid process run $\procrun =  f_1, \dots, f_n $
% by decoding with respect to $\nu$
% the variable sets 
% $S_i$ (to get the transitions taken), $M_{i,p}$ (to get the markings), and $X_{i,v}$ (to get the state variable assignments) for every instant $i$, 
% as described in Step (1).
% Moreover, we use the known correspondence between edit distance and alignments~\cite{NeedlemanW70}
% to reconstruct an alignment $\gamma = \gamma_{m,n}$ of $\logtrace$ and $\procrun$.

%\vspace{-2mm}

\section{Evaluation}
\label{sec:evaluation}
%\vspace{-2mm}
\noindent\textbf{Implementation.}
Our tool \thetool is implemented in Python as a branch of \cocomot.%
\footnote{See \texttt{\url{https://github.com/bytekid/cocomot/tree/object-centric}}.}
\thetool uses the SMT solver  %\zzz~\cite{deMouraB08} or 
\yices \textsf{2}~\cite{Dutertre14} as backend.
As inputs, the tool takes an \mynet and an object-centric 
event log. 
The former we represent as a \texttt{pnml} file with minimal adaptations, namely every place has an attribute \emph{color}, which is a tuple of types in $\Sigma$; and every arc has an \emph{inscription} attribute that is a list of typed variables as in \defref{inscription}.
The %object-centric 
event log is represented as a \texttt{jsonocel} file, in the format used in~\cite{LissAA23,AalstB20}.
The tool provides a simple command-line interface, and it outputs a textual representation of the decoded run of the \mynet, as well as the alignment and its cost.
The encoding in our implementation is as described in~\secref{encoding}, apart from
some simple optimizations that introduce additional SMT variables for subexpressions that occur multiple times.

\smallskip
\noindent\textbf{Experiments.}
We followed \cite{LissAA23} and used the BPI 2017 event data~\cite{BPI2017} restricted to the most frequent 50\% of activities. This resulted in 715 variants after applying \texttt{pm4py}'s variant filtering as in \cite{LissAA23}. The resulting traces have 3 to 23 events (12 on average), and 2 to 11 objects (5 on average). The \mynet for the experiments was obtained by augmenting their object-centric Petri net with arc inscriptions as remarked in \remref{OPIfromtheirnet}.
The evaluation was performed single-threaded but distributed on a 12-core
Intel i7-5930K 3.50GHz machine with 32GB of main memory.
% \todo{important: decide on whether initial marking is fixed or not}

\figref{results} summarizes the results of the experiments, showing the runtime in seconds in relationship to the number of events, the number of objects, and the cost of the optimal alignment. Note that the scale on the y-axis is logarithmic, the x-axis is linear.
The runtime is on average about 200 seconds, with maximal runtime of 7900 seconds.
While the tool of~\cite{LissAA23} is faster in comparison, the overhead seems acceptable given that \thetool performs a considerably more complex task. Experiment data and examples can be found on github.

\begin{figure}[t]
\begin{tikzpicture}
\begin{scope}
\begin{semilogyaxis}[
xmin = 0, xmax = 23,
ymin = 0.1, ymax = 23000,
ytick={0.1,1,100,10000},
yticklabels={0.1,1,100, $10^4$},
width = 49mm,
height = 49mm,
mark size=.8pt,
font=\tiny,
ylabel= {time in seconds},
xlabel={number of events}
]
\addplot[only marks,red!80!black, mark=*] file[skip first] {data/events_vs_time.dat};
\end{semilogyaxis}
\end{scope}
\begin{scope}[xshift=40mm]
\begin{semilogyaxis}[
xmin = 0, xmax = 12,
ymin = 0.1, ymax = 23000,
ytick={0.1,1,100,10000},
yticklabels={0.1,1,100, $10^4$},
width = 49mm,
height = 49mm,
mark size=.8pt,
font=\tiny,
xlabel={number of objects}
]
\addplot[only marks,blue!80!black, mark=o] file[skip first] {data/objects_vs_time.dat};
\end{semilogyaxis}
\end{scope}
\begin{scope}[xshift=80mm]
\begin{semilogyaxis}[
xmin = 0, xmax = 43,
ymin = 0.1, ymax = 23000,
ytick={0.1,1,100,10000},
yticklabels={0.1,1,100, $10^4$},
width = 49mm,
height = 49mm,
mark size=.8pt,
font=\tiny,
xlabel={alignment cost}
]
\addplot[only marks,green!80!black, mark=o] file[skip first] {data/distance_vs_time.dat};
\end{semilogyaxis}
\end{scope}
% \node[scale=.8] at (2, -.5) {(a)};
% \node[scale=.8] at (6, -.5) {(b)};
% \node[scale=.8] at (10, -.5)  {(c)};
\end{tikzpicture}
\caption{Run time related to the number of events, objects, and alignment cost.\label{fig:results}}
\end{figure}
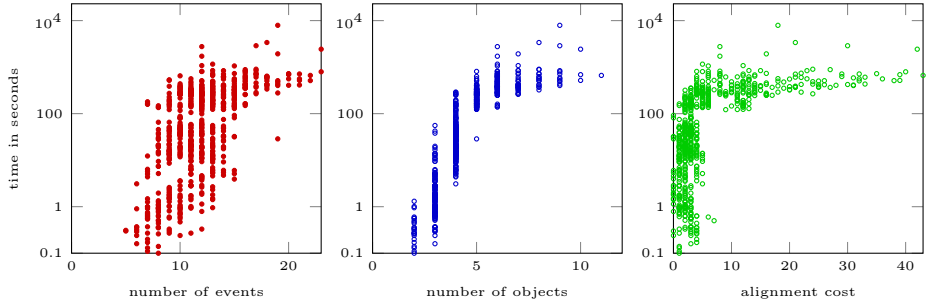
% \todo{what to say here?}

\section{Conclusion}
\label{sec:conclusion}

In this work we have defined the new formalism of \mynets for object-centric processes, supporting the essential modelling features in this space, as described in \secref{related}. We have then defined alignment-based conformance for this model.
% We presented a new conformance checking approach for object-centric Petri nets with identifiers, based on an SMT encoding.
% Along the way, we introduce a new variant of Petri net called \mynets, that combine arcs with multiplicity as in \cite{AalstB20}, with explicit object manipulation as in~\cite{RVFE10,GhilardiGMR22}; in fact \mynets subsume \cite{AalstB20} and to some extent proclets with synchronization~\cite{Fahland19}.
We can thus answer the research questions posed in the introduction affirmatively:
\begin{inparaenum}
\item[(1)] \mynets constitute a succinct, Petri net-based formalism for object-centric processes that supports object identity, arcs with multiplicity, and synchronization,
\item[(2)] we formally defined the conformance checking problem for \mynets, and 
\item[(3)] SMT encodings are a feasible, operational technique to implement this notion of conformance checking, as witnessed by our experiments.
\end{inparaenum}
% By means of examples, we show that \mynets are sensitive to object identity and admit synchronization, overcoming important limitations of \cite{LissAA23}.
% In experiments, we also show that our respective implementation is feasible.

In future work, we want to extend \mynets along two directions. First, we aim at supporting data and data-aware conditions, exploiting the fact that our approach seamlessly integrates with \textsf{CoCoMoT}. %This also provides the basis for further extending the approach with different forms of uncertainty in the log \cite{FGMR22}.
% In contrast, a respective extension of an approach based on synchronous product nets like \cite{LissAA23} seems much harder to achieve. 
Second, while exact synchronization is so far tackled only when computing alignment, we want to make it part of the model. This is non trivial: \mynets need to be equipped with a form of wholeplace operation, which requires universal quantification in the SMT encoding.

\bibliographystyle{splncs04}
\bibliography{references}

\newpage
\appendix

\section{Bounds on Optimal Alignment}
\label{sec:bounds}

\begin{lemma}
\label{lem:bounds:precise}
Let  $\NN$ be \anet and $T_X=\tup{E_X,D_X}$ a trace graph with optimal alignment $\Gamma$.
Let $m=\sum_{e\in E_X} |\projobj(e)|$ the number of object occurrences in $E_X$, and 
$c = \sum_{i=1}^n |\dom(b_i)|$ the number of object occurrences in some run  $\rho$ of $\NN$, with $\rho_v=\tup{\tup{t_1,b_1}, \dots \tup{t_n,b_n}}$.
Then $\Gamma|_\mod$ has 
at most $(|E_X|+c+m)(k+1)$ moves if $\NN$ has no $\nu$-inscriptions, and at most $(|E_X|+3c+2m)(k+1)$ otherwise, where $k$ is the longest sequence of silent transitions without $\nu$-inscriptions in $\NN$.
Moreover, $\Gamma|_\mod$ has at most $2c+m$ object occurrences in non-silent transitions.
\end{lemma}
\begin{proof}
Let $\Gamma_0$ be the alignment that consists of $|E_X|$ log moves with the events in $E_X$ and $n$ model moves using the transition firings in $\pi$.
By assumption, $\cost(\Gamma_0)=m+c$, and obviously $\cost(\Gamma) \leq \cost(\Gamma_0)$.
There can be at most $|E_X|$ synchronous moves of cost 0.
In addition, there can be up to $c+m$ non-silent model moves in $\Gamma$ to arrive at a cost of at most $c+m$, as every such move has cost at least 1.
These make $|E_X|+c+m$ moves in total; but in between each of these, as well as before and afterwards, there can be at most $k$ silent moves.

Now suppose that the OPI contains $\nu$-inscriptions. 
By the above assumption, every object is mentioned in at least one non-silent transition.
For an object $o$ that occurs in $\Gamma\restrmod$, we can distinguish two situations:
if it appears in a synchronous move $((a_\log,O_\log),(a_\mod,O_\mod))$ such that also $o \in O_\log$, then $o$ does not contribute to $\cost(M)$;
but if it appears in a synchronous move with $o \not\in O_\log$ or in a model move, it contributes 1 to $\cost(\Gamma)$.
Since there are $c$ objects in $E_X$, $c$ objects can occur in the former kind of situations (with no cost).
In the latter kind of situations, at most $c+m$ objects can occur, as every occurrence incurs a cost of 1, and $\cost(\Gamma) \leq \cost(\Gamma_0)$.
In the worst case where all these objects are distinct, $2c+m$ objects occur overall.
As above, there are at most $|E_X|$ synchronous moves and $c+m$ non-silent model moves.
There can be at most $2c+m$ silent model moves with outgoing $\nu$ edges, as there are at most that many objects.
These make $|E_X|+3c+2m$ moves in total; but in between each of these, as well as before and afterwards, there can be at most $k$ silent moves,
which shows the claim.
\qed
\end{proof}

\section{Encoding}

In this section we detail the encoding of \secref{encoding}.

\noindent\textbf{Variables.}
We start by fixing the set of variables used to represent the (unknown) model run and alignment:
\begin{compactenum}[(a)]
\item Transition variables $\transvar_j$ of type integer for all $1\leq j\leq n$ to identify 
the $j$-th transition in the run. To this end, we enumerate the transitions as $T = \{t_1, \dots, t_L\}$, and add the constraint $\bigwedge_{j=1}^n 1\leq \transvar_j \leq L$,
with the semantics that $\transvar_j$ is assigned value $l$ iff the $j$-th transition in $\rho$ is $t_l$.
\item To identify the markings in the run, we use marking variables $\markvar_{j,p,\vec o}$ of type boolean for every time point $0\leq j\leq n$, every place $p\in P$, and every vector $\vec o$ of objects with elements in $O$ such that $\coloring(\vec o)=\coloring(p)$. The semantics is that $\markvar_{j,p,\vec o}$ is assigned true iff $\vec o$ occurs in $p$ at time $j$.
\item To keep track of which objects are used by transitions of the run, we use object variables $\objvar_{j,k}$ of type integer for all $1\leq j\leq n$ and $0\leq k \leq K$ with the constraint $\bigwedge_{j=1}^n 1\leq \objvar_{j,k} \leq |O|$. 
The semantics is that if $\objvar_{j,k}$ is assigned value $i$ then, if $i>0$ the $k$-th object involved in the $j$-th transition is $o_i$, and if $i=0$ then the $j$-th transition uses less than $k$ objects.
\end{compactenum}
In addition, we use the following variables to represent alignment cost:
\begin{compactenum}
\item[(d)] Distance variables $\distvar_{i,j}$ of type integer for every $0\leq i\leq m$ and $0\leq j\leq n$, their use will be explained later.
\end{compactenum}
\smallskip

\noindent\textbf{Constraints.}
We use the following constraints on the variables defined above:
\begin{compactenum}[(1)]
\item 
\emph{Initial markings}.
We first need to ensure that the first marking in the run $\rho$ is initial.
By the expression $[\vec o \in M(p)]$ we abbreviate $\top$ if an object tuple $\vec o$ occurs in the $M(p)$, and $\bot$ otherwise.
\begin{align}
\label{eq:phi:initial}
\tag{$\varphi_{\mathit{init}}$}
\textstyle
\bigvee_{M \in M_{init}}
\bigwedge_{p\in P} \bigwedge_{\vec o \in \vec O_{\coloring(p)}} \markvar_{0,p, \vec o} = [\vec o \in M(p)] 
\end{align}
\item 
\emph{Final markings.} 
Next, we state that after at most $n$ steps, but possibly earlier, a final marking is reached.
\begin{align}
\label{eq:phi:final}
\tag{$\varphi_{\mathit{fin}}$}
\bigvee_{0 \leq j \leq n}
\bigvee_{M \in M_{final}}
\textstyle\bigwedge_{p\in P} \bigwedge_{\vec o \in \vec O_{\coloring(p)}} \markvar_{j,p, \vec o} = [\vec o \in M(p)] 
\end{align}
 \item 
\emph{Moving tokens.}
Transitions must be enabled, and tokens are moved by transitions. We encode this as follows:
\newcommand\prodtoken{\mathit{prod}\text{-}{tok}}
\newcommand\constoken{\mathit{cons}\text{-}{tok}}
\begin{align}
\bigwedge_{j=1}^n \bigwedge_{l=1}^L \transvar_j=l \rightarrow 
&\bigwedge_{p \in \pre{t_l}\setminus \post{t_l}} \bigwedge_{\vec o\in \vec O_{\coloring(p)}} (\constoken(p,t_l,j,\vec o) \to  \markvar_{j-1,p,\vec o} \wedge \neg\markvar_{j,p,\vec o}) \wedge {}
\notag\\
&\bigwedge_{p \in \pre{t_l}\cap \post{t_l}} \bigwedge_{\vec o\in \vec O_{\coloring(p)}} (\constoken(p,t_l,j,\vec o) \to  \markvar_{j-1,p,\vec o}) \wedge {}
\notag \\
&\bigwedge_{p \in \post{t_l}} \bigwedge_{\vec o\in \vec O_{\coloring(p)}} 
(\prodtoken(p,t_l,j,\vec o) \to  \markvar_{j,p,\vec o})
\tag{$\varphi_{\mathit{move}}$}
\label{eq:phi:tokens}
\end{align}
where $\constoken(p,t,j,\vec o)$ expresses that token $\vec o$ is consumed from $p$ in the $j$th transition which is $t$, and similarly
$\prodtoken(p,t,j,\vec o)$ expresses that token $\vec o$ is produced.
Formally, this is encoded as follows: first, suppose that $\inflow(p,t) = (v_1, \dots, v_h)$ is a non-variable flow, and let
$(k_1, \dots, k_h)$ be the object indices for $t$ of $v_1, \dots, v_h$. 
Then
$\constoken(p,t,j,\vec o) := (\bigwedge_{i=1}^{h} \objvar_{j,k_i} = id(\vec o_i))$, i.e., we demand that every variable used in the transition is instantiated to the respective object in $\vec o$.
If $\inflow(p,t) = (V_1, \dots, v_h)$ is a variable flow, suppose without loss of generality that $V_1\in \listvarset$.
Variable $V_1$ can be instantiated by multiple objects in a transition firing.
This is also reflected by the fact that there are several (but at most $K$) inscription indices
corresponding to instantiations of $V_1$, say $\ell_1, \dots, \ell_x$.
For $k_i$ as above for $i>1$, we then set
$\constoken(p,t,j,\vec o) := (\bigwedge_{i=2}^{h} \objvar_{j,k_i} = id(\vec o_i)) \wedge \bigvee_{i=1}^x \objvar_{j,\ell_i} = id(\vec o_1)$.
The shorthand $\prodtoken$ is encoded similarly, using $\outflow(t,p)$.
 \item 
\emph{Tokens that are not moved by transitions remain in their place.}
\begin{align}
\notag
\bigwedge_{j=1}^{n+1} 
\bigwedge_{p \in P}
\bigwedge_{\vec o\in \vec O_{\coloring(p)}} (\markvar_{j-1,p,\vec o} \leftrightarrow \markvar_{j,p,\vec o}) \vee &\bigvee_{t_l \in \post{p}} (\transvar_j\eqn l \wedge \constoken(p,t,j,\vec o) ) \vee {}\\
\label{eq:phi:inertia}
\tag{$\varphi_{\mathit{rem}}$}
& \bigvee_{t_l \in \pre{p}} (\transvar_j\eqn l \wedge \prodtoken(p,t,j,\vec o) ) 
\end{align}
 \item 
\emph{Transitions use objects of suitable type.}
To this end, recall that every transition can use at most $K$ objects, which limits instantiations of template inscriptions.
For every transition $t\in T$, we can thus enumerate the objects used by it from 1 to $K$.
However, some of these objects may be unused. We use the shorthand $\mathit{needed}_{t,k}$ to express this: $\mathit{needed}_{t,k} = \top$ if the $k$-th object is necessary for transition $t$ because it occurs in a simple inscription, and $\bot$ otherwise.
Moreover, let $\mathit{ttype}(t,k)$ be the type of the $k$-th object used by transition $t$.
Finally, we denote by $O_\sigma$ the subset of objects in $O$ of type $\sigma$.
\begin{equation}
\label{eq:phi:type}
\tag{$\varphi_{\mathit{type}}$}
\bigwedge_{j=1}^n \bigwedge_{l=1}^L \transvar_j=l \rightarrow \bigwedge_{k=1}^K \left ((\neg [\mathit{needed}_{t_l,k}] \wedge \objvar_{j,k}=0) \vee \bigvee_{o \in O_{\mathit{ttype}(t_l,k)}} \objvar_{j,k}=id(o)\right)
\end{equation}
 \item 
\emph{Objects that instantiate $\nu$-variables are fresh.}
We assume in the following constraint that $tids_\nu$ is the set of all $1 \leq l \leq L$ such that $t_l$ has an outgoing $\nu$-inscription, and that every such $t_l$ has only one outgoing $\nu$-inscription $\nu_t$, and we assume w.l.o.g. that in the enumeration of objects of $t$, $\nu_t$ is the first object. However, the constraint can be easily generalized to more such inscriptions.
\begin{equation}
\label{eq:phi:fresh}
\tag{$\varphi_{\mathit{fresh}}$}
\bigwedge_{j=1}^n \bigwedge_{l\in tids_\nu} \bigwedge_{o\in O_{\otype(\nu_t)}} \transvar_j=l \wedge \objvar_{j,1} = id(o)  \rightarrow 
(\bigwedge_{p\in P} \bigwedge_{\vec o \in \vec O_{\coloring(p)}, o \in \vec o} \neg \markvar_{j-1,p,\vec o})
\end{equation}
\end{compactenum}

\noindent\textbf{Encoding alignment cost.}
Similar as in~\cite{FelliGMRW23,BoltenhagenCC21}, we encode the cost of an alignment as the edit distance with respect to suitable penalty functions $P_=$, $P_M$, and $P_L$.
Given a trace graph
$T_X = (E_X, D_X)$, let
\begin{equation}
\logtrace = \tup{e_1, \dots, e_m}
\label{eq:ordered:trace}
\end{equation}
be an enumeration of all events in $E_X$ such that
$\projtime(e_1) \leq \dots \leq \projtime(e_m)$.
Let the penalty expressions $[P_L]_i$, $[P_M]_{j}$, and $[P_=]_{i,j}$
be as follows, for all $1\leq i \leq m$ and $1 \leq j \leq n$:
\begin{align*}
[P_L]_{i} &= |\projobj(e_i)| \qquad\qquad [P_=]_{i,j} = 
ite(\mathit{is\_labelled}(j, \projact(e_i)), 0, \infty)\\
[P_M]_{j} &=
ite(\mathit{is\_labelled}(j, \tau)%\vee \rlvar <j
, 0, \Sigma_{k=1}^K ite(\objvar_{j,k} \neq 0, 1, 0))
\end{align*}
where $\mathit{is\_labelled}(j,a)$ expresses that the $j$-the transition has label $a\in \activities \cup \{\tau\}$, which can be done by taking $is\_labelled(j,a):= \bigvee_{l\in T_{idx}(a)}\transvar_j = l$
where $T_{idx}(a)$ is the set of transition indices with label $a$, i.e., the set of all $l$ with $t_l \in T$ such that $\ell(t_l) = a$.

Using these expressions, one can encode the edit distance as in \cite{FelliGMRW23,BoltenhagenCC21}:
\begin{equation}
 \label{eq:delta}
\begin{array}{rl@{\qquad}rl@{\qquad}rl@{\qquad\quad}r}
\distvar_{0,0} &= 0 &
\distvar_{{i+1},0} &= [P_L] + \distvar_{i,0} &
\distvar_{0,{j+1}} &= [P_M]_{j+1} + \distvar_{0,j} 
\\[1ex]
\distvar_{i+1,j+1} &\multicolumn{5}{l}{=
\min (
[P_=]_{i+1, j+1} + \distvar_{i,j},\ 
[P_L] + \distvar_{i,j+1},\ 
[P_M]_{j+1} + \distvar_{i+1,j})}
\end{array}
\tag{$\varphi_\delta$}
\end{equation}

\noindent\textbf{Solving.}
We abbreviate 
$\varphi_{\mathit{run}} = \varphi_{\mathit{init}} \wedge \varphi_{\mathit{fin}} \wedge \varphi_{\mathit{move}} \wedge \varphi_{\mathit{rem}} \wedge \varphi_{\mathit{type}} \wedge \varphi_{\mathit{fresh}}$ and
use an SMT solver to obtain a satisfying assignment $\alpha$ for the 
following constrained optimization problem: 

\begin{align*}
\label{eq:constraints}
\varphi_{\mathit{run}} \wedge
\varphi_{\delta}
\text{\quad minimizing\quad }\distvar_{m,n}
\tag{$\Phi$}
\end{align*}

\noindent\textbf{Decoding.}
From an assignment $\alpha$ satisfying \eqref{eq:constraints}, we next define a run $\rho_\alpha$ and an alignment $\Gamma_\alpha$.
First, we note the following:
From the result of \secref{bounds}, we can obtain a number $M$ such that $M$ is the maximal number of objects used to instantiate a list variable in the model run and alignment. By convention, we may assume that in the enumeration of objects used in the $j$th transition firing, $\objvar_{j, |O|-M+1}, \dots, \objvar_{j, |O|}$ are those instantiating a list variable, if there is a list variable in $\invars{t_{\alpha(\transvar_j)}} \cup \outvars{t_{\alpha(\transvar_j)}}$.

We assume the set of transitions $T=\set{t_1, \dots, t_L}$ is ordered as $t_1, \dots, t_L$ in some arbitrary but fixed way that was already used for the encoding. 

\begin{definition}[Decoded run]
\label{def:decoded:run}
For $\alpha$ satisfying \eqref{eq:constraints}, let the \emph{decoded process run} be
$\rho_\alpha =  \tup{f_1, \dots, f_n}$ such that for all $1\leq j \leq n$,  $f_j = (\widehat t_j, b_j)$, where $\widehat t_j  = t_{\alpha(\transvar_j)}$ and $b_j$ is defined as follows:
Assuming that $\invars{t_{\alpha(\transvar_j)}} \cup \outvars{t_{\alpha(\transvar_j)}}$ is ordered as $v_1, \dots, v_k$ in an arbitrary but fixed way that was already considered for the encoding,
we set $b_j(v_i) = \alpha(\objvar_{j,i})$ if $v_i \in \varset$, and 
$b_j(v_i) = [O_{\alpha(\objvar_{j,|O|-M+1})}, \dots, O_{\alpha(\objvar_{j,|O|-M+z})}]$ if $v_i \in \varset$, where $0 \leq z < M$ is maximal such that $\alpha(\objvar_{j,|O|-M+z}) \neq 0$.
\end{definition}

At this point, $\rho_\alpha$ is actually just a sequence; we will show below that it is indeed a process run of $\NN$.
Next, given  a satisfying assignment $\alpha$ for \eqref{eq:constraints}, we define an alignment of the log trace $T_X$ and the process run $\rho_\alpha$.

\begin{definition}[Decoded alignment]
\label{def:decoded:alignment}
For $\alpha$ satisfying \eqref{eq:constraints}, $\rho_\alpha =  \tup{f_1, \dots, f_n}$ as defined above, and $\logtrace$ as in \eqref{eq:ordered:trace},
consider the sequence of moves $\gamma_{i,j}$ recursively defined as follows:
\begin{align*}
\gamma_{0,0} &=\epsilon \qquad
\gamma_{i+1,0}= \gamma_{i,0} \cdot \tup{e_{i+1}, \SKIP} \qquad
\gamma_{0,j+1}= \gamma_{0,j} \cdot \tup{\SKIP, f_{j+1}} \\
\gamma_{i+1,j+1} &= 
\begin{cases}
\gamma_{i,j+1} \cdot \tup{e_{i+1}, \SKIP} &
 \text{ if }\alpha(\distvar_{i+1,j+1}) = \alpha([P_L] + \distvar_{i,j+1}) \\
\gamma_{i+1,j} \cdot \tup{\SKIP, f_{j+1}} &
 \text{ if otherwise }\alpha(\distvar_{i+1,j+1}) = \alpha([P_M]_{j+1} + \distvar_{i+1,j}) \\
\gamma_{i,j} \cdot \tup{e_{i+1}, f_{j+1}} &
 \text{ otherwise}
\end{cases}
\end{align*}
Given $\gamma_{i,j}$, we define a 
graph $\Gamma(\gamma_{i,j})=\tup{C,B}$ of moves as follows: the node set $C$ consists of all moves in $\gamma_{i,j}$, and there is an edge $\tup{\tup{q,r}, \tup{q',r'}} \in B$ if either $q\neq \SKIP$, $q' \neq \SKIP$ and there is an edge $q \to q'$ in $T_X$, or
if $r\neq \SKIP$, $r' \neq \SKIP$, $r=f_h$, and $r'=f_{h+1}$ for some $h$ with $1 \leq h < n$.
Finally, we define the decoded alignment as $\Gamma(\alpha) := \Gamma(\gamma_{m,n})$.
\end{definition}
In fact, as defined, $\Gamma(\alpha)$ is just a graph of moves, it yet has to be shown that it is a proper alignment. This will be done in the next section.
\smallskip

\noindent\textbf{Correctness.}
In the remainder of this section, we will prove that $\rho_\alpha$ is indeed a run, and $\Gamma(\alpha)$ is an alignment of $T_X$ and $\rho_\alpha$. We first show the former:

\begin{lemma}
\label{lem:decode:run}
Let $\NN$ be \anet, $T_X$ a log trace and $\alpha$ a solution to $\eqref{eq:constraints}$. Then $\rho_\alpha$ is a run of $\NN$.
\end{lemma}
\begin{proof}
We define a sequence of markings $M_0, \dots, M_n$.
Let $M_j$, $0 \leq j \leq n$, be the marking such that $M_j(p)= \{\vec o \mid \vec o \in \vec O_{\coloring(p)} \text{ and } \alpha(\markvar_{j,p,\vec o})=\top\}$.
Then, we can show by induction on $j$ that for the process run $\rho_j = \tup{f_1, \dots, f_j}$ 
it holds that  $M_0\goto{\rho_j} M_j$.
\begin{itemize}
\item[\textit{Base case.}] If $n=0$, then $\rho_0$ is empty, so the statement is trivial. 
\item[\textit{Inductive step.}] Consider $\rho_{j+1} = \tup{f_1, \dots, f_{j+1}}$ and
suppose that for the prefix
$\rho' = \tup{f_1, \dots, f_j}$
it holds that $M_0 \goto{\rho'} M_j$.
We have $f_{j+1} = (\widehat t,b)$ and $\widehat t=t_i$ for some $i$ such that $1\,{\leq}\,i\,{\leq}\,|T|$ with $\alpha(\transvar_j) = i$.
First, we note that $b$ is a valid binding: as $\alpha$ satisfies \eqref{eq:phi:type}, it assigns a non-zero value to all $\objvar_{j,k}$ such that $v_k \in \invars{t_i} \cup \outvars{t_i}$ that are not of list type (and hence needed), and by $\eqref{eq:phi:type}$, the unique object $o$ with $id(o)=\alpha(\objvar_{j,k})$ has the type of $v_k$.
Similarly, $b$ assigns a list of objects of correct type to a variable in
$(\invars{t_i} \cup \outvars{t_i}) \cap \listvarset$, if such a variable exists.
Moreover, \eqref{eq:phi:fresh} ensures that variables in $(\invars{t_i} \cup \outvars{t_i}) \cap \nuvarset$ are instantiated with objects that did not occur in $M_j$.

Since $\alpha$ is a solution to $\eqref{eq:constraints}$, it satisfies \eqref{eq:phi:tokens}, so that $t_i$ is enabled in $M_n$.
As $\alpha$ satisfies \eqref{eq:phi:inertia}, the new marking $M_{j+1}$ contains only either tokens that were produced by $t_i$, or tokens that were not affected by $t_i$. Thus, $M_j\goto{f_{j+1}} M_{j+1}$, which concludes the induction proof.

\end{itemize}
Finally, as $\alpha$ satisfies 
\eqref{eq:phi:initial} and \eqref{eq:phi:final}, it must be that $M_0=M_I$ and the last marking must be final, so $\rho_\alpha$ is a run of $\NN$.
\qed
\end{proof}

\begin{theorem}
Given \anet $\NN$, trace graph $\tracenet_X$, and satisfying assignment $\alpha$ to $\eqref{eq:constraints}$, $\Gamma(\alpha)$ is an optimal alignment of $\tracenet_X$ and the run $\rho_\alpha$  with cost $\alpha(\distvar_{m,n})$.
\end{theorem}
\begin{proof}
By \lemref{decode:run}, $\rho_\alpha$ is a run of $\NN$.
We first note that $[P_=]$, $[P_L]$, and $[P_M]$ are correct encodings of 
$P_=$, $P_L$, and $P_M$ from \defref{cost}, respectively. For $P_L$ this is clear.
For $P_=$, $\mathit{is\_labelled}(j,a)$ is true iff the value of $\transvar_j$ corresponds to a transition
that is labeled $a$. If the labels match, cost $0$ is returned, otherwise $\infty$.
For $P_M$, the case distinction returns cost $0$ if the $j$th transition is silent; otherwise, the expression 
$\Sigma_{k=1}^K ite(\objvar_{j,k} \neq 0, 1, 0)$ 
counts the number of objects involved in the model step, using the convention that if fewer than $k$ objects are involved in the $j$th transition then $\objvar_{j,k}$ is assigned 0.

Now, let $d_{i,j} = \alpha(\distvar_{i,j})$, for all $i$, $j$ such that $0 \leq i \leq m$ and $0 \leq j \leq n$.
Let again $\logtrace=\tup{e_1, \dots, e_m}$ be the sequence ordering the nodes in $T_X$ as in \eqref{eq:ordered:trace}. 
Let $T_X|_i$ be the restriction of $T_X$ to the node set $\{e_1, \dots, e_i\}$.
We show the stronger statement that $\Gamma(\gamma_{i,j})$ is an optimal alignment of $T_X|_i$ and $\rho_\alpha|_j$ with cost $d_{i,j}$, by induction on $(i,j)$. % with respect to the lexicographic order.

\begin{itemize}
\item[\textit{Base case.}] If $i\,{=}\,j\,{=}\,0$, 
then $\gamma_{i,j}$ is the trivial, empty alignment of an empty log trace and an empty process run, which is clearly optimal with cost
$d_{i,j}\,{=}\,0$, as defined in \eqref{eq:delta}.
\item[\textit{Step case.}]
If $i\,{=}\,0$ and $j\,{>}\,0$, then $\gamma_{0,j}$ is a sequence of model moves 
$\gamma_{0,j}=\tup{(\SKIP, f_{1}),\ldots, (\SKIP, f_{j}) }$ according to \defref{decoded:alignment}.
Consequently, $\Gamma(\alpha)=\Gamma(\gamma_{0,j})$ has edges $(\SKIP, f_{h}),\ldots, (\SKIP, f_{h+1})$ for all $h$, $1\leq h < j$, which is a valid and optimal alignment of the empty log trace and $\rho_\alpha$.
By \defref{cost}, the cost of $\Gamma(\alpha)$ is the number of objects involved in 
non-silent transitions of $f_1, \dots, f_j$, which coincides with
$\alpha([P_M]_1 + \dots + [P_M]_j)$, as stipulated in \eqref{eq:delta}.
\item[\textit{Step case.}]
If $j\,{=}\,0$ and $i\,{>}\,0$, then $\gamma_{i,0}$ is a sequence of log moves 
$\gamma_{i,0}=\tup{(e_{1},\SKIP),\ldots, (e_{j}, \SKIP) }$ according to \defref{decoded:alignment}.
Thus, $\Gamma(\alpha)=\Gamma(\gamma_{i,0})$ is a graph whose log projection
coincides by definition with $T_X|_i$.
By \defref{cost}, the cost of $\Gamma(\alpha)$ is the number of objects involved in 
$e_1, \dots, e_i$, which coincides with
$\alpha([P_L]_1 + \dots + [P_L]_j)$, as stipulated in \eqref{eq:delta}.
\item[\textit{Step case.}] 
If $i\,{>}\,0$ and  $j\,{>}\,0$, 
$d_{i,j}$ must be the minimum of
$\alpha([P_=]_{i, j}){+}d_{i-1,j-1}$,
$\alpha([P_L]){+}d_{i-1,j}$, 
and $\alpha([P_M]_{j}){+}d_{i,j-1}$.
We can distinguish three cases:
\begin{compactitem}
\item
Suppose $d_{i,j} = \alpha([P_L]_i) + d_{i-1,j}$.
By \defref{decoded:alignment},
we have $\gamma_{i,j} = \gamma_{i-1,j} \cdot \tup{e_{i}, \SKIP}$.
Thus, $\Gamma(\gamma_{i,j})$ extends $\Gamma(\gamma_{i-1,j})$ by a node
$\tup{e_{i}, \SKIP}$, and edges to this node as induced by $T_X|_i$.
By the induction hypothesis, $\Gamma(\gamma_{i-1,j})$ is a valid and optimal alignment
of $T_X|_{i-1}$ and $\rho_\alpha|_{j}$ with cost $d_{i-1,j}$.
Thus $\Gamma(\gamma_{i,j})$ is a valid alignment of $T_X|_i$ and $\rho_\alpha|_{j}$,
because the log projection coincides with $T_X|_i$ by definition.
By minimality of the definition of $d_{i,j}$, also $\Gamma(\gamma_{i,j})$ is optimal.
\item
Suppose $d_{i,j} = \alpha([P_M]_{j}) + d_{i,j-1}$.
By \defref{decoded:alignment},
we have $\gamma_{i,j} = \gamma_{i,j-1} \cdot \tup{\SKIP, f_j}$.
By the induction hypothesis, $\Gamma(\gamma_{i,j-1})$ is a valid and optimal alignment
of $T_X|_{i}$ and $\rho_\alpha|_{j-1}$ with cost $d_{i,j-1}$.
Thus, $\Gamma(\gamma_{i,j-1})$ must have a node $\tup{r,f_{j-1}}$, for some $r$.
The graph $\Gamma(\gamma_{i,j})$ extends $\Gamma(\gamma_{i,j-1})$ by a node
$\tup{\SKIP, f_{j}}$, and an edge $\tup{r,f_{j-1}} \to \tup{\SKIP, f_{j}}$.
Thus, $\Gamma(\gamma_{i,j})$ is a valid alignment for $T_X|_{i}$ and $\rho_\alpha|_{j}$, and by minimality it is also optimal.
\item
Let $d_{i,j} = \alpha([P_=]_{i, j}) + d_{i-1,j-1}$.
By \defref{decoded:alignment},
we have $\gamma_{i,j} = \gamma_{i-1,j-1} \cdot \tup{e_i, f_j}$.
By the induction hypothesis, $\Gamma(\gamma_{i-1,j-1})$ is an optimal alignment
of $T_X|_{i-1}$ and $\rho_\alpha|_{j-1}$ with cost $d_{i-1,j-1}$.
In particular, $\Gamma(\gamma_{i-1,j-1})$ must have a node $\tup{r,f_{j-1}}$, for some $r$. The graph $\Gamma(\gamma_{i,j})$ extends $\Gamma(\gamma_{i-1,j-1})$ by a node
$\tup{e_1, f_{j}}$, an edge $\tup{r,f_{j-1}} \to \tup{e_i, f_{j}}$, as well as edges to $\tup{e_1, f_{j}}$ as induced by $T_X|_i$.
Thus $\Gamma(\gamma_{i,j})$ is a valid alignment of $T_X|_i$ and $\rho_\alpha|_{j}$,
because the log projection coincides with $T_X|_i$ by definition, and the model projection has the required additional edge.
By minimality of the definition of $d_{i,j}$, also $\Gamma(\gamma_{i,j})$ is optimal.
\end{compactitem}
%
% As $[P_=]$, $[P_L]$, and $[P_M]$ are correct encodings of  $P_=$, $P_L$, and $P_M$, respectively,
% by definition of a distance-based cost function (Def. \ref{def:distance}),
% $d_{i,j}$ is thus the optimal alignment cost of  $\logtrace|_i$ and $\rho_\alpha|_j$,
% and it is clear from the construction of $\gamma_{i,j}$ is an alignment that has exactly this cost.
\end{itemize}
For the case $i=m$ and $j=n$, we obtain that $\Gamma(\alpha)=\Gamma(\gamma_{m,n})$ is an optimal alignment of $\tracenet_X$ and $\rho_\alpha$ with cost $d_{m,n}=\alpha(\distvar_{m,n})$.
\qed
\end{proof}

\end{document}